\documentclass[a4paper,UKenglish,cleveref, autoref, thm-restate]{lipics-v2021}
\hideLIPIcs  %

\graphicspath{{./figs/}}%

\nolinenumbers

\usepackage{todonotes}
\usepackage{complexity}
\usepackage{tcolorbox}
\usepackage{xspace}

\tcbset{%
  width=\textwidth,
  halign=justify,
  center,
  colback=white,
  arc=0pt,
  outer arc=0pt,
  boxrule=1pt,
  colbacktitle=white,
  coltitle=black,
  titlerule=0pt,
  left=5pt,
  right=1pt,
  top=1pt,
  bottom=1pt
}

\crefname{observation}{Observation}{Observations}
\Crefname{observation}{Observation}{Observations}

\newcommand{\crn}{\operatorname{cr}}
\newcommand{\ppdcrn}{\operatorname{ppdcr}}
\newcommand{\crnsl}{\operatorname{\overline{cr}}}
\newcommand{\ftpcrn}[1]{\operatorname{cr}_{#1}}
\newcommand{\tw}{\operatorname{tw}}

\newcommand{\ceil}[1]{\left\lceil #1 \right\rceil}

\newcommand{\mypar}[1]{{\smallskip\sffamily\bfseries #1.\xspace}}

\AtBeginEnvironment{tcolorbox}{\small}

\title{Computing crossing numbers with topological and geometric restrictions} %

\author{Thekla Hamm}{TU Eindhoven, The Netherlands}{t.l.s.hamm@tue.nl}{https://orcid.org/0000-0002-4595-9982}{acknowledges support by the Austrian Science Fund FWF (Project J4651-N).}

\author{Fabian Klute}{Universitat Politècnica de Catalunya, Spain}{fabian.klute@upc.edu}{https://orcid.org/0000-0002-7791-3604}{F. K. is supported by María Zambrano grant 2022UPC-MZC-94041 funded by the Spanish Ministry of Universities and the European Union (NextGenerationEU) and by grant PID2023-150725NB-I00 funded by MICIU/AEI/10.13039/501100011033.}

\author{Irene Parada}{Universitat Politècnica de Catalunya, Spain}{irene.parada@upc.edu}{https://orcid.org/0000-0003-3147-0083}{I. P. is a Serra H\'unter Fellow and partially supported by grant PID2023-150725NB-I00 funded by MICIU/AEI/10.13039/501100011033.}

\authorrunning{T. Hamm, F. Klute and I. Parada} %

\Copyright{Thekla Hamm, Fabian Klute, and Irene Parada} %



\category{} %

\relatedversion{} %

\acknowledgements{This work started during the 12th Crossing Numbers Workshop in Strobl, Austria. 
We thank the organizers and participants for an inspiring atmosphere. 
In particular, we thank Bruce Richter and Birgit Vogtenhuber for early discussions on the problem. 
}%

\begin{document}

\maketitle

\begin{abstract}
Computing the crossing number of a graph is one of the most classical problems in computational geometry.
Both it and numerous variations of the problem have been studied, and overcoming their frequent computational difficulty is an active area of research.
Particularly recently, there has been increased effort to show and understand the parameterized tractability of various crossing number variants.
While many results in this direction use a similar approach, a general framework remains elusive.
We suggest such a framework that generalizes important previous results, and can even be used to show the tractability of deciding crossing number variants for which this was stated as an open problem in previous literature.
Our framework targets variants that prescribe a partial predrawing and some kind of topological restrictions on crossings.
Additionally, to provide evidence for the non-generalizability of previous approaches for the partially crossing number problem to allow for geometric restrictions, we show a new more constrained hardness result for partially predrawn rectilinear crossing number.
In particular, we show \W-hardness of deciding \textsc{Straight-Line Planarity Extension} parameterized by the number of missing edges.
\end{abstract}

\clearpage

\section{Introduction}
The \emph{crossing number} $\crn(G)$ of a graph $G$, i.e., 
the smallest number $\crn(G)$ such that there exists a drawing
of it with at most $\crn(G)$ intersection between the interiors of the edges,
is among the most widely studied graph parameters in graph algorithms and graph theory.
The extensive survey by Schaefer~\cite{Schaefer2013} on the subject collects 
almost 800 references.

The classical computational complexity of
deciding whether a given graph $G$ has crossing number $k$
is well-understood.
The problem was already shown to be \NP-complete by Garey and Johnson~\cite{garey1983crossing}, and even \APX-hard~\cite{DBLP:journals/dcg/Cabello13} in general.
Since then, many variants and restrictions of the problem have been studied and, most of the time, shown to be \NP-hard as well~\cite{DBLP:journals/siamcomp/CabelloM13,DBLP:journals/jct/Hlineny06a,DBLP:journals/algorithmica/PelsmajerSS11}.
However, when we consider the complexity depending on \(k\) positive algorithmic results are possible:
Graphs with crossing number zero, i.e.\ \emph{planar graphs}, are long known to be recognizable in polynomial time, see for example~\cite{DBLP:journals/jacm/HopcroftT74}.
Beyond this, it was shown~\cite{DBLP:journals/jcss/Grohe04,DBLP:conf/stoc/KawarabayashiR07} that one can decide in time \(f(k)|V(G)|^{\mathcal{O}(1)}\) for some computable function \(f\) whether the crossing number of a graph \(G\) is at most \(k\), i.e.\ deciding the crossing number of a graph is in \FPT\ with respect to its natural parameterization.
Numerous results on crossing number variants can be obtained using the same approach, with variant-specific adaptations~\cite{DBLP:journals/jgaa/BannisterE18,DBLP:conf/icalp/GanianHKPV21,DBLP:conf/compgeom/HammH22,DBLP:journals/algorithmica/PelsmajerSS11}.
Motivated by this, very recently M\"unch and Rutter~\cite{munch2024parameterized} suggested a framework abstracting this.
It targets crossing number variants which prohibit drawings in which edges with shared endpoints cross in a certain combinatorial way.

The aforementioned results use Courcelle's theorem which results in a high parameter-dependence of the time complexity of the implied algorithms.
Recently, Colin de Verdière and Hliněný published a manuscript~\cite{DBLP:journals/corr/abs-2410-00206}
also obtaining \FPT-algorithms by the natural parameterization for several crossing number variants while avoiding the use of Courcelle's theorem.
Instead, they reduce to the problem of embedding graphs into 2-complexes, a known approach to solve the crossing number problem~\cite{DBLP:conf/esa/VerdiereM21}.

None of the previous approaches can distinguish the topological behavior of crossings in relation to each other.
Indeed, 
allowing for forbidding configurations distinguishing a certain type of topological behavior of crossings was explicitly left as an open problem in~\cite{munch2024parameterized}.
For example, computing the crossing number when requiring the drawing to be 
pseudolinear requires such topological properties.
A drawing of a graph is called \emph{pseudolinear} if 
the edges can be extended to a pseudoline arrangement.\footnote{Informally, a pseudoline arrangement is a set of simple curves in the plane, each of which extends on both sides to infinity, which pairwise cross in a unique point.}
Deciding if a given drawing is pseudolinear can be done
by checking for a set of obstructions in the drawing~\cite{DBLP:journals/jocg/ArroyoBR21}.
These obstructions inherently contain information about the topological behavior of the crossings, e.g., a cycle of crossings leaving a certain set of vertices on one side, compare \Cref{fig:plobstructions}.

\begin{figure}
    \centering
    \includegraphics[page=1]{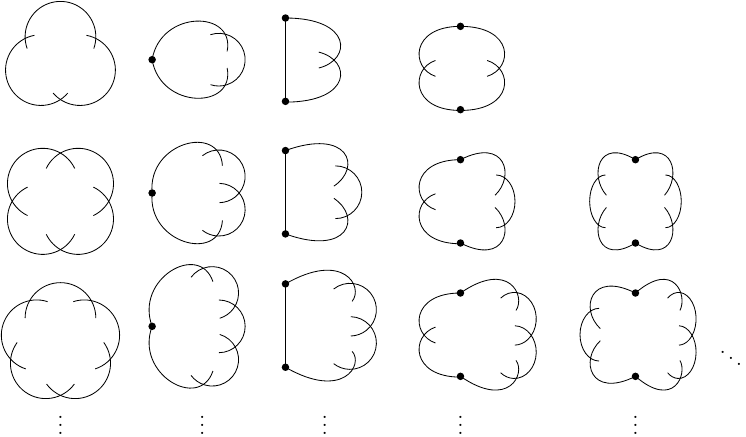}
    \caption{The obstructions for pseudolinear drawing, figure based on \cite[Figure~2]{DBLP:journals/jocg/ArroyoBR21}.}%
    \label{fig:plobstructions}
\end{figure}

An even more ambitious goal would be to also capture geometric restrictions on the drawing.
One of the most fundamental geometric restrictions for graph drawings requires straight-line edge drawings.
This gives rise to the \emph{rectilinear crossing number} $\crnsl(G)$.
Computing it is known to be an $\exists\mathbb R$-complete~\cite{DBLP:journals/dcg/Bienstock91} problem,
lending credence to the idea that it is algorithmically harder.
Another hint in this direction is that extending a given straight-line 
plane partial drawing of a graph to a straight-line plane drawing of the whole graph is \NP-hard~\cite{DBLP:journals/ijfcs/Patrignani06}, 
while the combinatorial setting is linear-time solvable~\cite{DBLP:journals/talg/AngeliniBFJKPR15}.

Besides considering the solution value as a parameter, the perspective of \emph{drawing extension} allows for the study of the complexity of drawing problems when a large part of the drawing is \emph{predrawn} via parameterizing by the size of the non-predrawn part of the instance.
This has led to algorithmic results for several classically hard drawing problems, e.g.,~\cite{DBLP:conf/mfcs/EibenGHKN20,DBLP:conf/icalp/EibenGHKN20,DBLP:conf/icalp/GanianHKPV21,DBLP:conf/compgeom/HammH22}.

\mypar{Contribution} We present a framework to compute the 
partially predrawn crossing number in \FPT-time
under a wide array of drawing restrictions.
For this we develop the notion of \emph{topological crossing patterns} and
show that we can decide the partially predrawn crossing number of a graph under 
the drawing restrictions prohibiting small and few patterns in \FPT-time
parameterized by the resulting crossing number variant.
Crucially, our patterns are also able to capture topological aspects of the forbidden structures.
This makes our framework more general than the one proposed by M\"unch and Rutter~\cite{munch2024parameterized} and
the usage of forbidden patterns makes it a more directly applicable toolbox than
the unified approach by Colin de Verdière and Hliněný~\cite{DBLP:journals/corr/abs-2410-00206}
which requires the reduction of the problem at hand into a less intuitive problem.
Moreover, it seems infeasible to apply their result to the setting where a graph can be accompanied by a predrawing of a subgraph that is to be respected.

We demonstrate the strength of our framework by showing 
that computing the crossing number of so-called fan-planar graphs is in \FPT,
resolving and explicit open problem from~\cite{munch2024parameterized}.
Moreover, we show how to compute the partially predrawn pseudolinear crossing number in \FPT-time
when parameterizing by the number of new crossings which requires us to handle 
forbidden patterns whose concrete realizations in the drawing might be of unbounded size.

Entering the world of straight-line drawings we sharpen the complexity-theoretic bounds
known for different types of partially predrawn crossing number problems.
We prove that deciding the partially predrawn rectilinear crossing number 
parameterized by the number of edges that are to be inserted is $W[1]$-hard,
even when the partial and the final drawing are plane.
This is in stark contrast to other extension problems,
that are usually in \FPT\ parameterized by the number of to-be-inserted edges~\cite{DBLP:journals/jgaa/ChimaniH23,DBLP:conf/icalp/EibenGHKN20,DBLP:conf/icalp/GanianHKPV21}.

\mypar{Algorithm overview}
For our algorithmic framework, we follow the same general approach that was used in the first 
\FPT-result for deciding the crossing number, namely the unrestricted one.
This approach works in two steps:
Reducing to the case of bounded treewidth, and then expressing the existence of a drawing witnessing the crossing number value at most \(k\) as an MSO-formula whose length is bounded in \(k\) to which we can then apply Courcelle's theorem.

In more detail, in the first step we need to decide that we have a no-instance, bounded treewidth, or a way to reduce the graph.
One subtlety is that in the case that we reduce the graph, we do so in a slightly different way than previous works~\cite{DBLP:journals/jcss/Grohe04,DBLP:conf/compgeom/HammH22,munch2024parameterized}; specifically, by removing an edge rather than contracting anything.
Contracting may lead to new connectivity, and hence is not robust for forbidding topological crossing patterns.
In the framework proposed by M\"unch and Rutter this lead to necessitating the requirement that in their forbidden crossing patterns, non-crossing vertices need at least one crossing neighbor.

In the second step, we use the fact that we now can assume bounded treewidth to apply Courcelle's theorem.
For this we need to express the existence of a targeted drawing in a crossing-number-bounded-length MSO-formula.
Two things have already been done frequently in literature:
Expressing the existence of a drawing with at most \(k\) crossings in MSO, and forbidding the occurrence of a small size of small subgraphs or subdivided subgraphs in a host graph.
However, in our result, the fact that we want to speak about the topological behavior of such subgraphs adds a challenging new layer.
We overcome this difficulty by describing the choices relevant for the existence of an embedding of the prospective planarized graph via existentially quantified variables.
For this it is crucial that we constrain the forbidden patterns that we allow in our algorithmic framework in a way that the number of relevant choices is bounded in our parameter.
In particular, relevant choices will only be relative to a small number of vertices in the prospective planarization. 
Doing this allows us to speak about the embedding choices when checking for the existence of subgraphs which under these choices have a certain topological behavior.

\mypar{Outline}
We introduce necessary notation and concepts in \Cref{sec:prelims}.
In \Cref{sec:problem} we define the notion of topological crossing patterns and
present the main algorithmic result of the paper and 
its implications.
\Cref{sec:bd-tw,sec:algorithm} are dedicated to the
proof of our \FPT-time algorithm to solve the 
partially predrawn crossing number with forbidden topological crossing patterns.
\Cref{sec:hardness} contains our hardness result.
We conclude in \Cref{sec:conclusion} with open problems.

\section{Preliminaries}
\label{sec:prelims}
We use standard notation from graph theory for undirected graphs~\cite{DBLP:books/daglib/0030488}.
We also assume familiarity with basic concepts from parameterized complexity theory~\cite{DBLP:books/sp/CyganFKLMPPS15,DBLP:series/txcs/DowneyF13} and
graph drawing as it relates to crossing numbers~\cite{Schaefer2013}.
Specifically, we assume that the reader is familiar with Courcelle's theorem and its applications~\cite{DBLP:journals/iandc/Courcelle90,Downey2013} and $W[1]$-hardness reductions~\cite{DBLP:series/txcs/DowneyF13}.

For a graph $G$ its vertices are denoted with $V(G)$ and its edges with $E(G)$.
A drawing $\mathcal G$ of $G$ is a mapping of each vertex $v\in V(G)$ to distinct points $\mathcal G(v)$ in $\mathbb R^2$
and of each edge $e = uv \in E(G)$ with vertices $u,v \in V(G)$ to a simple curve $\mathcal G(e)$ that starts at $\mathcal G(u)$ and ends at $\mathcal G(v)$
points corresponding to its incident vertices.
When clear from context we identify vertices and edges with their respective drawings.
If not otherwise specified we assume that in a drawing of a graph
two edges share at most one point, 
if two edges share an interior point they cross properly in that point (not tangentially),
and no more than two edges cross in any point apart from vertices.
The \emph{planarization} $\mathcal G^\times$ of a drawing $\mathcal G$ of a graph $G$ is obtained 
from $\mathcal G$ by replacing each crossing with a vertex and every then subdivided curve by an edge.

We define a \emph{partial drawing} of a graph $G$ as the drawing of any subgraph $H$ of $G$.
Then, a \emph{partially predrawn} graph is just the tuple $(G,\Gamma)$ where
$\Gamma$ is a drawing of some subgraph of $G$.
Given a partially predrawn graph \((G,\Gamma)\), a drawing of \(G\) \emph{extends} \(\Gamma\) if it is `equivalent' to \(\Gamma\) when restricted to the graph drawn by \(\Gamma\).
There are different ways to interpret the equivalence of drawings.
For us two ways are relevant:
\emph{topological equivalence} -- where all rotations of edge drawings around vertices in the planarizations of both drawings have to be equal and the same sets of vertices have to be face boundaries in the planarizations of both drawings --
and \emph{geometric equivalence} -- where the actual embedding functions of vertices to points and edges to curves in \(\mathbb{R}^2\) have to be equal.

For \(i \in \mathbb{N}\), we denote by \(H_i\) the \emph{hexagonal grid} of radius \(i\).
A \emph{topological embedding} \(h : H_i \to G\) injectively maps each vertex of \(H_i\) to a vertex of \(G\) and the edges of \(H_i\) to internally vertex-disjoint paths in \(G\) between the \(h\)-images of their endpoints.
The existence of topological embeddings of (hexagonal) grids into a graph is related to its treewidth by what are commonly called grid minor or excluded grid theorems. We use the following variant.

\begin{theorem}[\cite{DBLP:journals/jct/RobertsonS86}]
\label{thm:grid-minor-thm}
    Let \(s \geq 1\).
    There is an \(w \geq 1\) that only depends on \(s\) and an algorithm that given an input graph \(G\) correctly outputs \(\tw(G) \leq w\) or computes a topological embedding \(h : H_s \to G\) in time \(f(s)|V(G)|\) for some computable function \(f\).
\end{theorem}

In the context of crossing numbers, it has been important to find not only topological embeddings of grids but topological embeddings of grids that together with everything attached to them are planar.
Such grids are \emph{flat}.
While we will not be able to use known results on the existence of flat topological embeddings of grids because we need to adapt them for the partially predrawn setting, it is useful to already introduce previous terminology that we will fall back on later in \Cref{sec:bd-tw} along with some notation.
We denote by \(h(H_i)\) the graph given by \((h(V(H_i)) \cup \bigcup_{e \in E(H_i)}V(h(H_i)),\bigcup_{e \in E(H_i)}E(h(H_i)))\)
A \emph{\(h(H_i)\)-component} of \(G\) is a connected component of \(G - h(H_i)\) together with all edges between that component and \(h(H_i)\).
A \(h(H_i)\)-component is \emph{proper} if it is connected to at least one vertex of \(h(H_i)\).
We will use \(h(H_i)^+\) to denote the union of \(h(H_i)\) with all its proper components.

\subsection{Specific Problem Definitions}
We will apply our algorithmic framework to two specific crossing number problems which can be viewed as instantiations of the following problem with different drawing styles \(\Sigma\).

\begin{tcolorbox}
{{\sc Partially Predrawn \(\Sigma\) Crossing Number}\hfill \textbf{Parameter:} \(k\)}

\noindent\textbf{Input:} A \(\Sigma\)  partially predrawn graph \((G,\Gamma)\) where \(G\) is conntected and an integer \(k\).\\
\noindent\textbf{Question:} Is there a topological drawing extension of \(\Gamma\) to a \(\Sigma\) drawing \(\mathcal{G}\) of \(G\) such that \(\crn(\mathcal{G}) - \crn(\Gamma) \leq k\)?
\end{tcolorbox}

Constrained to instances in which \(\Gamma=\emptyset\), we simply speak of \textsc{\(\Sigma\) Crossing Number}.
In the special case that \(\Sigma\) places no restriction on the drawing style, we are left with the \textsc{Partially Predrawn Crossing Number} problem, for which we call the smallest \(k\) for which a partially predrawn graph \((G,\Gamma)\) is a yes-instance, the \emph{partially predrawn crossing number} of \((G,\Gamma)\) and denote it as \(\ppdcrn(G,\Gamma)\).

We will consider \emph{fan-planarity}
and \emph{pseudolinearity} as drawing style requirements prescribed by \(\Sigma\).
Both of these have known characterizations through forbidden subdrawings (\Cref{fig:fanobstructions,fig:plobstructions}).

\begin{figure}
    \centering
    \includegraphics[page=3]{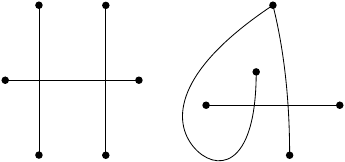}
    \caption{The two obstructions characterizing fan-planar drawings.}
    \label{fig:fanobstructions}
\end{figure}
We consider the original definition of fan-planarity which 
characterizes the class of fan-planar graphs as those that admit
a drawing not containing the obstructions shown in \Cref{fig:fanobstructions}~\cite{Kaufmann2022}.\footnote{In the literature this class is also called weak fan-planar graphs~\cite{DBLP:conf/gd/CheongFKPS23}.}
That is, no drawing of a subgraph can be topologically equivalent to any of them.
Moreover, every fan-planar drawing of a graph with a minimum number of crossings can be assumed to be simple~\cite{DBLP:journals/jgaa/KlemzKRS23}.
Hence, \textsc{Fan-Planar Crossing Number} can equivalently be stated as asking for the existence of a topological drawing extension of \(\Gamma\) to \(G\) with at most \(k\) new crossings and none of the forbidden configurations in \Cref{fig:fanobstructions}.

\begin{tcolorbox}
{{\sc Fan-Planar Crossing Number}\hfill \textbf{Parameter:} \(k\)}

\noindent\textbf{Input:} A connected graph \(G\) and an integer \(k\).\\
\noindent\textbf{Question:} Is there a fan-planar drawing \(\mathcal{G}\) of \(G\) such that \(\crn(\mathcal{G}) \leq k\)?
\end{tcolorbox}

It is known that pseudolinear drawings can be characterized by the absence of a family of forbidden patterns~\cite{DBLP:journals/jocg/ArroyoBR21}, i.e.\ no drawing of a subgraph can be topologically equivalent to any of the obstructions shown in \Cref{fig:plobstructions}.
Hence, \textsc{Partially Predrawn Pseudolinear Crossing Number} can equivalently be stated as asking for the existence of a topological drawing extension of \(\Gamma\) to \(G\) with at most \(k\) new crossings and none of the above forbidden patterns.

\begin{tcolorbox}
{{\sc Partially Predrawn Pseudolinear Crossing Number}\hfill \textbf{Parameter:} \(k\)}

\noindent\textbf{Input:} A pseudolinearly partially predrawn graph \((G,\Gamma)\) where \(G\) is conntected and an integer \(k\).\\
\noindent\textbf{Question:} Is there a topological drawing extension of \(\Gamma\) to a pseudolinear drawing \(\mathcal{G}\) of \(G\) such that \(\crn(\mathcal{G}) - \crn(\Gamma) \leq k\)?
\end{tcolorbox}

In geometric settings, geometric rather than topological drawing extensions are more natural.
In view of the known \NP-hardness of deciding whether a geometric straight-line extension without crossings exists~\cite{DBLP:journals/ijfcs/Patrignani06}, it becomes important to consider more restrictive parameters than the solution value.
Specifically, we will consider the following parameterized problem.\footnote{Straight-line drawings are sometimes called \emph{rectilinear} drawings or \emph{geometric graphs} in the literature. However, the term \emph{rectilinear drawings} sometimes refers to orthogonal drawings in the area of Graph Drawing. To avoid confusion, we call the drawings straight-line. But since the \emph{rectilinear crossing number} problem has a long history under that name, we name our problems following this terminology.}

\begin{tcolorbox}
{{\sc Partially Predrawn Rectilinear Crossing Number}\hfill \textbf{Parameter:} \(|E(G) - E(\Gamma)|\)}

\noindent\textbf{Input:} A straight-line partially predrawn graph \((G,\Gamma)\) where \(G\) is connected and an integer \(k\).\\
\noindent\textbf{Question:} Is there a geometric drawing extension of \(\Gamma\) to a straight-line drawing \(\mathcal{G}\) of \(G\) such that \(\crn(\mathcal{G}) - \crn(\Gamma) \leq k\)?
\end{tcolorbox}

We will not apply our algorithmic framework to \textsc{Partially Predrawn Rectilinear Crossing Number}, but show hardness instead.
For the presentation of our reduction we are even able to avoid any crossing,
albeit allowing disconnected instances.

\begin{tcolorbox}
{{\sc Straight-Line Planarity Extension} \hfill \textbf{Parameter:} \(|E(G) - E(\Gamma)|\)}

\noindent\textbf{Input:} A plane straight-line partially predrawn graph \((G,\Gamma)\).\\
\noindent\textbf{Question:} Is there a geometric drawing extension of \(\Gamma\) to a plane straight-line drawing of \(G\)?
\end{tcolorbox}

\section{Forbidding Topological Crossing Patterns}
\label{sec:problem} %

Because in general we consider drawing extension problems, our patterns will need to distinguish each type of vertex (real, crossing or of a specific color) also based on whether it is predrawn in the instance or not.
Moreover, to allow speaking about topological aspects of forbidden crossing patterns, we allow patterns to be partially predrawn themselves:
A \emph{topological crossing pattern} is a partially predrawn planar graph \((P,\Pi)\) such that
\begin{align*}
    V(P) = V_R \, \dot{\cup} \, V_C = V_A \, \dot{\cup} \, V_\Phi
\text{ and } E(P) = E_R \, \dot{\cup} \, E_P = E_A \, \dot{\cup} \, E_\Phi\end{align*}
possibly with a coloring \(c\) of \(V_R\) and \(E_P\).

Here the partitions of \(V(P)\) and \(E(P)\) are into original -- \(V_R\) -- and planarized crossing vertices -- \(V_C\), into abstract non-embedded -- \(V_A\) -- and instance-predrawn vertices -- \(V_\Phi\) (these do not necessarily have to be predrawn in the pattern as well) -- and into original -- \(E_R\) -- and contracted edges -- \(E_P\) -- and into abstract non-embedded -- \(E_A\) -- and predrawn edges -- \(E_\Phi\), respectively.
The size of a pattern is its number of vertices.
We require that both endpoints of any edge in \(E_\Phi\) are in \(V_\Phi\) (this corresponds to the fact that endpoints of instance-predrawn edges need to be predrawn).
Justifying the term `crossing pattern' we also require that for each connected component \(Q\) of \(P\), \(V(Q) \cap V_C \neq \emptyset\).
Similarly, we require that every connected component of the graph drawn by \(\Pi\) contains a vertex of \(N_P[V_C]\), i.e.\ the predrawing of a pattern always speaks about how (embedded structures containing) crossings or endpoints of crossing edges are separated by the drawings of other parts of a pattern.
We also enforce some additional restrictions, whose importance 
we point out where they become relevant in the proof of our algorithmic framework.

Our first additional restriction, the \emph{side specification constraint},
requires that the maximum vertex degree in
\((\{v \in V(\Pi) \mid \exists e \in E_A \ v \in e\},E(\Pi))\)
is at most \(2\) (this still allows us to capture many topological constraints on crossings).
Further, we constrain all endpoints of edges in \(E_P\) to be in \(V_C \cap V_A\).
This technical assumption, denoted as the \emph{crossing locality assumption}, will be necessary for our arguments to reduce to bounded treewidth in \Cref{sec:bd-tw}.

A \emph{realization} \(\tilde{P}\) arises from a drawing extension of \((P,\Pi)\) by replacing the edges in \(E_P\) by pairwise disjoint paths of arbitrary length \(\geq 1\) (possibly single edges) such that the rotation of edges around each vertex \(v \in V(P) \cap V(\tilde{P})\) is equal in the embedding of \(P\) and \(\tilde{P}\), if in the latter edges to subdivision vertices (i.e.\ vertices in \(V(\tilde{P}) \setminus V(P)\)) are replaced by the endpoint of the corresponding edge in \(E_P\) that is not \(v\).
If \(e \in E_P\) has color \(c(e)\) we let the color of all subdivision vertices for \(e\) also be \(c(e)\).
Note that the size of a realization of a pattern can be arbitrarily large compared to the size of the pattern; indeed this is necessary for applications like \textsc{Partially Predrawn Pseudolinear Crossing Number} (\Cref{cor:pseudolinear}).

\(P\) \emph{occurs} in a drawing extension \(\mathcal{G}\) of a possibly vertex-colored partially predrawn graph \((G,\Gamma)\) if \(\mathcal{G}^\times\) contains an embedded graph that is isomorphic to some realization \(\tilde{P}\) of \(P\) by an isomorphism that
(i) only maps planarized crossings of \(\mathcal{G}\) to vertices in \(V_C\) and only maps vertices in \(V(G)\) to vertices in \(V(\tilde{P}) \setminus V_C\) and
(ii) only maps vertices from \(V(\Gamma^\times)\) to \(V_\Phi\) and edges from \(E(\Gamma^\times)\) to \(E_\Phi\).
Moreover, if vertex colors are present, we require the isomorphism to respect this coloring, i.e. to map a vertex of \(G\) with color \(c(v)\) to each \(v \in V(\tilde{P})\).
We call \(\mathcal{G}\) \emph{\(P\)-free} if \(P\) does not occur in \(\mathcal{G}\).
For a set of topological crossing patterns \(\mathbb{P}\), we say \(\mathcal{G}\) is \(\mathbb{P}\)-free if \(\mathcal{G}\) is \(P\)-free for all \(P \in \mathbb{P}\).
We denote by \(\ftpcrn{\mathbb{P}}(G,\Gamma)\) the minimum number of crossings involving edges in \(E(G) \setminus E(\Gamma)\) in any \(\mathbb{P}\)-free drawing extension of \((G,\Gamma)\).

Lastly,
we will require a type of closure from \(\mathbb{P}\).
Like the crossing pattern locality assumption, this will come into play in \Cref{sec:bd-tw}.
For a set of topological crossing patterns \(\mathbb{P}\) and \(\ell \in \mathbb{N}\) we say that it is \emph{\(\ell\)-crossing contraction safe} if for all \(P \in \mathbb{P}\) and all patterns \(P'\) that arise from \(P\) by an arbitrary subdivision of the set of edges in \(E_P\) using new pairwise distinct vertices in \(V_C \cap V_A\) and a possible arbitrary identification of pairs of some of these new vertices such that after this \(|V_C \cap V_A| \leq \ell\), \(P' \in \mathbb{P}\).
We show the following result:
\begin{theorem}
\label{thm:ftpcr-fpt}
    Given a possibly vertex-colored partially predrawn graph \((G,\Gamma)\), \(k \in \mathbb{N}\), and a \(k\)-crossing contraction safe set of topological crossing patterns \(\mathbb{P}\), deciding \(\ftpcrn{\mathbb{P}}(G,\Gamma) \leq k\) is in \FPT\ parameterized by \(k + |\mathbb{{P}}| + \max\{|P| \mid P \in \mathbb{P}\}\).
\end{theorem}

\begin{figure}
    \centering
    \includegraphics[page=2]{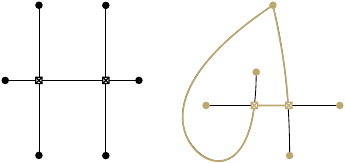}
    \hfil
    \includegraphics[page=12]{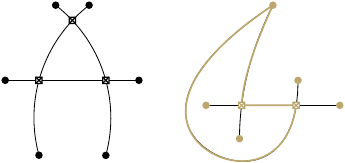}
    \caption{The two basic forbidden patterns for fan-planar crossing number (left). 
    The two other versions how the two patterns might appear in a planarized drawing.}
    \label{fig:fanpattern}
\end{figure}

\begin{remark}
    \Cref{thm:ftpcr-fpt} generalizes the algorithmic framework (Theorem~1) from \cite{munch2024parameterized} 
    in three ways: 
(i) we can forbid patterns whose realization in a potential drawing spans an arbitrary number of vertices (in particular, potentially unbounded in the crossing number variant), 
(ii) we can handle (partially) topologically embedded patterns, and 
(iii) we can handle a given partial drawing which we need to extend.
\end{remark}

The case of no pattern being forbidden and a tailored adaption using one pattern to handle $k$-planar drawings was studied in previous work by Hamm and Hliněný~\cite{DBLP:conf/compgeom/HammH22}.
Their results can now be viewed in a more unified way in our framework.

\begin{remark}
    \Cref{thm:ftpcr-fpt} generalizes the non-oriented version of the main algorithmic theorem (Theorem~1.1) from \cite{DBLP:conf/compgeom/HammH22} which shows fixed-parameter tractibility of \textsc{Partially Predrawn Crossing Number},
    as well as the adaptation of this result to \textsc{Partially Predrawn \(c\)-Planar Crossing Number} \cite[Theorem~1.3]{DBLP:conf/compgeom/HammH22}.
    This is the case because by definition \(\ppdcrn(G,\Gamma) = \ftpcrn{\emptyset}(G,\Gamma)\), and because restricting to \(c \leq k\)-planarity means considering a path with \(c+1\) vertices in \(V_C\) and all its edges in \(E_R\) as single forbidden pattern.
\end{remark}

Before we give the formal proof of \Cref{thm:ftpcr-fpt},
we demonstrate how it can be applied.
The added power of our framework compared to previous results allows us to show that deciding the fan-planar and the (partially predrawn) pseudolinear crossing number are each in \FPT\ with respect to the number of crossings.
The former solves an open question from~\cite{munch2024parameterized}.

\begin{corollary}
    \label{cor:fanplanar}
    \textsc{Fan-Planar Crossing Number} is in \FPT.
\end{corollary}

We derive the set of patterns from the two basic ones shown in \Cref{fig:fanpattern}.
These are the planarizations of the obstructions to fan-planarity~\cite{DBLP:conf/gd/CheongFKPS23,Kaufmann2022}. 
For both basic patterns, we include every crossing vertex (marked by a cross) in $V_C$ and every
other vertex in $V_R$. All edges are contained in $E_R$ which implies that \(E_P = \emptyset\).
Moreover, all vertices and edges are added to $V_A$ and $E_A$.
Edges and vertices marked in brown are predrawn in the pattern.
Keep in mind that these might not be predrawn in the instance itself. 
In particular, for \textsc{Fan-Planar Crossing Number} no part of the instance is predrawn, but some part of some of our patterns are.
We now construct all patterns for $k$ crossings. 
First we add the two forbidden patterns shown on the right in \Cref{fig:fanpattern} and then
enumerate all possible ways how to subdivide the edges in the basic patterns with crossing vertices such that the resulting pattern has at most $k$ crossings.
The number of such graphs is bounded in $k$ only and
each graph has size linear in $k$.
Let $\mathbb P_f$ be the resulting set of forbidden patterns.
As $E_P$ is empty, it is relatively easy to check that the patterns in $\mathbb P_f$ conform to all requirements of \Cref{thm:ftpcr-fpt} and
since every pattern contains a subdivision of one of the forbidden configurations
we obtain the corollary.

For fan-planar crossing number we did not have to use all the capabilities of our framework. 
Most notably we did not need to add any edges to \(E_P\) when constructing the pattern graphs.

For our next result the basic patterns are shown in \Cref{fig:plabasicpatterns}.
Analogous to fan-planar drawings, they are derived from the known obstructions for pseudolinear drawings~\cite{DBLP:journals/jocg/ArroyoBR21}; see \Cref{fig:plobstructions}.
When given a non-empty partial predrawing the number of crossings already existing in the predrawing is not bounded in the value of our parameter, and an unbounded number of such crossings can contribute to large obstructions to pseudolinearity with few new crossings.
This is where we make use of \(E_P\) to still be able to forbid only few and small topological crossing patterns.

\begin{figure}
    \centering
    \includegraphics[page=2]{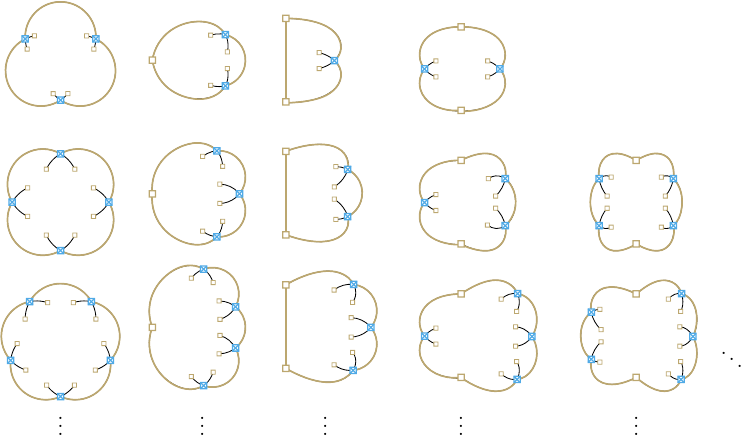}
    \caption{Basic patterns derived from the obstructions in \Cref{fig:plobstructions}. Cyan squares are crossings. Brown edges form the cycle and squares are crossings or vertices, both are predrawn in the pattern.}
    \label{fig:plabasicpatterns}
\end{figure}

\begin{corollary}
    \label{cor:pseudolinear}
    \textsc{(Partially Predrawn) Pseudolinear Crossing Number} is in \FPT.
\end{corollary}

Similarly to the situation with fan-planar crossing number we are going to construct the basic  patterns from the planarizations of the forbidden configurations shown in \Cref{fig:plobstructions}.
\Cref{fig:plabasicpatterns} shows the planarized versions of the forbidden configurations.
To aid the description we divide the planarization one forbidden configurations
into the \emph{cycle} which consists of the crossing vertices and the up to two vertices that
make up the cyclic boundary.
We add \emph{pendants} to every end of a curve in the patterns,
all these pendants lie on one side of the cycle.
Note that while crossing vertices in the planarization have to be crossings,
all other vertices in the patterns may be mapped to vertices or crossings by an isomorphism witnessing the realization of a forbidden pattern.

Moreover, since the choice of the outer face matters, 
we consider not just one forbidden configuration at a time, but each pattern will have to arise from the combination of three forbidden configurations.
Two of these are nesting each other 
and the third does not nest at least one of them.
In this way, no matter which face is chosen as the outer face, one of the obstructions is drawn 
such that all its pendants are drawn inside the corresponding cycle.
Conversely, for only two occurrences of forbidden configurations or three which all nest or are all pairwise non-nesting, there is always a choice of an outer face such that all pendants point towards it.

In case the partially prewdrawn graph in the given instance is empty, we proceed similarly as for \Cref{cor:fanplanar}.
In particular we only need to consider patterns with at most \(k\) crossing vertices, hence their size and number will be bounded.
The occurrence of larger forbidden patterns is already precluded by not a sufficient number of crossings being allowed (let us emphasize that this is only true if the patial predrawing is empty).
We enumerate all possibilities if a pendant is a vertex or a crossing.
Then, each crossing and vertex is added to $V_C$ and $V_R$, respectively.
All vertices and the edges on the cycle are predrawn in the pattern.
Finally, we generate the set of patterns $\mathbb P_p'$ by taking all such basic patterns with at most $k' \leq k$ crossing vertices and computing all possible subdivisions with up to $k-k'$ new crossing vertices.
Applying \Cref{thm:ftpcr-fpt} to the set of the above patterns as \(\mathbb{P}\) is sufficient to solve \textsc{Pseudolinear Crossing Number}.

The distinguishing aspect that comes into play when given a non-empty partial predrawing is that the number of crossings already existing in the predrawing is not bounded in the value of our parameter.
This is where we will need to make use of \(E_P\) to still be able to forbid only few and small topological crossing patterns.

Specifically, construct \(\mathbb{P}_p\) starting from \(\mathbb{P}_p'\) and add every possible 
topological pattern that arises from any \((P,\Pi)\) in $\mathbb P_p'$ via the following procedure.
Color an arbitrary subset 
of edges on the cycle with both endpoints in $V_C$ blue and add these edges to $E_P$.
Any edge adjacent to an edge now in \(E_P\) is no longer allowed to be predrawn in the pattern, so we remove it as well as the endpoints of any newly blue edge from \(\Pi\).
Following these rules, we generate all possible patterns with blue cycle edges.
Moreover, whenever an edge on the cycle is not sharing a vertex with a blue edge in $E_P$,
we add it to the edges predrawn in the pattern.
This means, we always predraw as many edges as possible on the cycle. 
To make the set \(k\)-crossing contraction safe, we also add all patterns derived 
by up to $k$ subdivisions of blue edges with new crossings and 
possibly identify pairs among them.
By construction, $\mathbb P_p$ conforms to all the conditions necessary for \Cref{thm:ftpcr-fpt}. 
We need to preprocess the partially predrawn graph $(G,\Gamma)$ on the input by 
coloring all vertices representing crossings in $\Gamma^\times$ as blue to ensure the desired behavior of blue edges in our patterns.

\begin{lemma}
    Let $\mathcal G$ be a drawing of a graph $G$, then 
    $\mathcal G$ is pseudolinear if and only if 
    $\mathcal G$ is $\mathbb P_p$-free.
\end{lemma}
\begin{proof}
    Assume that $\mathcal G$ is pseudolinear, but a pattern $P\in\mathcal P_p$ occurs in $\mathcal G$.
    If $P$ does not contain a blue edge this is impossible as then the three parts of $P$ correspond each to the subdivision of the graph underlying some obstruction for pseudolinear drawings and,
    moreover, at least one of them has to be embedded in $\mathcal G^\times$ such that it corresponds to the subdivision of an obstruction.

    Now let $P$ contain at least one blue edge.
    If $P$ occurs in $\mathcal G$, this means that the blue edges were replaced by paths of crossing vertices from $\mathcal G^\times$.
    Let $\tilde P$ be the corresponding realization.
    Consider the path $v_1,\ldots,v_a$ of $a \in O(|V|)$ crossing vertices that replaced a blue edge in $\tilde{P}$.
    Then the rotation of edges 
    around $v_2,\ldots,v_{a-1}$ is isomorphic to the one around $v_a$.
    By construction, this means that all the pendants attached to crossings on the cycle are drawn on one side of the disk bounded by the cycle.
    But then $\tilde{P}$ contains the subdivision of some obstruction (with possibly many crossings).

    In the other direction, assume that $\mathcal G$ is $\mathbb P_p$-free, but
    is in fact not pseudolinear.
    Then by~\cite{DBLP:journals/jocg/ArroyoBR21} $\mathcal G$ contains one of the obstructions $B$, such that it contains the minimal number of curve-parts that connect vertices and/or crossings.
    But then, by the way the patterns were constructed, there is a pattern in $\mathbb P_p$ whose
    underlying graph is isomorphic to $B^\times$.
    If $B^\times$ does not contain any crossing vertices corresponding to crossings in $\Gamma$, we are in fact 
    done, as we derive the corresponding pattern in $\mathbb P$ by predrawing the cycle edges and vertices $B^\times$.
    Now if $B^\times$ does contain some crossings from $\Gamma$, we may assume 
    that this crossing is on the cycle, as else the obstruction was not chosen with the minimal number of curve-parts connecting vertices and/or crossings.
    But then there exists a pattern containing exactly these paths of crossing vertices on the cycle as blue edges.
    Consequently, some pattern of $\mathbb P_p$ occurs in~$\Gamma$.
\end{proof}

\section{Reducing to Bounded Treewidth}
\label{sec:bd-tw}
We aim to find, in the case that we cannot correctly identify a no-instance of our problem or bound the treewidth of the graph in the given crossing number value, an edge of our instance that we can delete, such that the resulting smaller instance is equivalent to the original one.

For this, we find a sufficiently large grid minor \(h(H_i)\) which can be assumed to be drawn without any crossings in a hypothetical solution assuming sufficiently large treewidth.
Where for abstract graphs, without partial embeddings we would be able to use what is commonly referred to as flat grid theorem~\cite{DBLP:journals/jct/Thomassen97b}, we need its adaptation to the partially predrawn setting.
While this idea is already implicit in the \FPT-algorithm for partially predrawn crossing number~\cite{DBLP:conf/compgeom/HammH22}, we give a self-contained statement and short proof here.

Given a partially predrawn graph \((G,\Gamma)\), and a topological embedding \(h : H_i \to G\) for some \(i \in \mathbb{N}\), \(h(H_i)\) is called \emph{partially predrawn flat} if \(\ppdcrn(h(H_i^+) \cup \Gamma, \Gamma) = 0\).
\begin{theorem}
\label{thm:flat-grid}
    For any partially predrawn graph \((G,\Gamma)\) and \(r \geq 1\) there is an \(s \geq 1\) that depends only on \(r\) and \(\ppdcrn(G,\Gamma)\) such that the following holds.
    If \(h : H_s \to G\) is a topological embedding, then there is a subgrid \(H_r \subseteq H_s\) such that the restriction \(h\big|_{H_r}\) of \(h\) to \(H_r\) is partially predrawn flat.
    Such a subgrid can be found in time \(|V(G)|^{O(1)}\).
\end{theorem}
\begin{proof}
    We can generously choose \(s = (\ppdcrn(G,\Gamma) + 1) \cdot r\) to be able to find \(\ppdcrn(G,\Gamma) + 1\) pairwise disjoint subgrids of radius \(r\) each in \(H_s\).
    In a drawing realizing the value of \(\ppdcrn(G,\Gamma)\), at least one of these must be uncrossed together with all its proper components, because otherwise we would have at least one crossing in each and more than \(\ppdcrn(G,\Gamma)\) crossings overall.
    This means at least one of these subgrids is partially predrawn flat.
    To determine which one, we can go through them and apply the known algorithm for deciding partially predrawn planarity~\cite{DBLP:journals/talg/AngeliniBFJKPR15} of each of the disjoint grids with its proper components and \(Gamma\).
\end{proof}

Beside changing the number of crossings by deleting an edge as desired,
we need to worry about changing the existence of forbidden topological crossing patterns.
The arguments we make for this in the following proof rely on the crossing locality assumption as well as the \(k\)-crossing contraction safeness of the considered pattern sets.

\begin{lemma}
    \label{lem:wlog-uncrossed}
    Let \(k \in \mathbb{N}\), \(\mathbb{P}\) be a \(k\)-crossing contraction safe set of forbidden topological crossing patterns, and \((G,\Gamma)\) be a possibly vertex-colored partially predrawn graph.
    If \(\mathcal{G}\) is a drawing witnessing \(\ftpcrn{\mathbb{P}}(G,\Gamma) \leq k\) and \(h : H_s \to G\) is partially predrawn flat, then there is a drawing \(\mathcal{G}'\) also witnessing \(\ftpcrn{\mathbb{P}}(G,\Gamma) \leq k\) in which all edges in \(h(H_{i})^+\), where \(H_{i}\) consists of all but the \((k+1)(\max\{|P| \mid P \in \mathbb{P}\})\) outermost concentric cycles of \(H_{s}\), are uncrossed.
\end{lemma}
\begin{proof}
    Group the \((k+1)(\max\{|P| \mid P \in \mathbb{P}\})\)-outermost concentric cycles of \(H_s\) into disjoint groups of \(\max\{|P| \mid P \in \mathbb{P}\})\) consecutive cycles each, starting from the outermost.
    This leads to \(k+1\) groups.
    Hence, in any drawing \(\mathcal{G}\) witnessing \(\ftpcrn{\mathbb{P}}(G,\Gamma) \leq k\) there is one of these groups which together with all \(h(H_s)\)-components all of whose attachment vertices are on the topological embedding of one of the cycles of this group is uncrossed in \(\mathcal{G}\).

    Because $h(H_s)$ is partially predrawn flat, we can use the drawing of the \((\max\{|P| \mid P \in \mathbb{P}\} + 1)\)-outermost concentric cycle of \(h(H_s)\) to glue in a crossing-free drawing of all vertices and edges that are drawn inside that cycle by \(\mathcal{G}\) in a way that is consistent with \(\Gamma\). %
    We call the resulting drawing \(\mathcal{G}'\).
    
    It is clear that \(\mathcal{G}'\) does not have more crossings that \(\mathcal{G}\).
    To see why the above modification maintains \(\mathbb{P}\)-freeness, note that each vertex of the glued-in drawing has distance at least \(\max\{|P| \mid P \in \mathbb{P}\} + 1\) from any crossing vertex in the planarization of the drawing.
    By the crossing locality assumption, this implies that only monochromatic paths in the glued-in drawing both of whose endpoints are outside of the glued-in parts can be included in the occurrence of a pattern \(P \in \mathbb{P}\) in the considered drawing of \(G\).
    The same paths exist in \(\mathcal{G}^\times\) possibly interrupted by planarized crossing vertices.
    Because there are at most \(k\) crossings in \(\mathcal{G}\), if there is a pattern \(P \in \mathbb{P}\) occurring in \(\mathcal{G}'\), there is a pattern \(P'\) that occurs in \(\mathcal{G}\) whose containment in \(\mathbb{P}\) is implied by the \(k\)-crossing contraction safeness of \(\mathbb{P}\).

    This means that because \(\mathcal{G}\) is \(\mathbb{P}\)-free the same is true for \(\mathcal{G}'\) and thus all desired properties hold for \(\mathcal{G}'\).
\end{proof}

We are now ready to prove the lemma that allows us to reduce to bounded-treewidth instances.
\begin{lemma}
\label{lem:bd-tw}
    Let \(\mathbb{P}\) be a set of forbidden topological crossing patterns.
    There is a function \(g : \mathbb{N} \to \mathbb{N}\) and an algorithm that given an partially predrawn graph \((G,\Gamma)\) and \(k \in \mathbb{N}\)
    \begin{itemize}
        \item correctly outputs \(\ftpcrn{\mathbb{P}}(G,\Gamma) > k\), or
        \item correctly outputs \(\tw(G) \leq g(k,\max\{|P| \mid P \in \mathbb{P}\})\), or
        \item outputs an edge \(e \in E(G)\) such that \(\ftpcrn{\mathbb{P}}(G,\Gamma) \leq k\) if and only if \(\ftpcrn{\mathbb{P}}(G-e,\Gamma-e) \leq k\)
    \end{itemize}
    in time \(f(k,\max\{|P| \mid P \in \mathbb{P}\}) \cdot |V(G)|^{\mathcal{O}(1)}\).
\end{lemma}
\begin{proof}
    We can apply \Cref{thm:grid-minor-thm} with \(s = \bar{s}((k+1)\max\{|P| \mid P \in \mathbb{P} + 2\max\{|P| \mid P \in \mathbb{P}\}+3\},k)\) where the function \(\bar{s}\) maps \((r,c) \in \mathbb{N}^2\) to \(s \in \mathbb{N}\) as in \Cref{thm:flat-grid} for \(r\) and \(G\) with \(\ppdcrn(G,\Gamma) = c\).
    The \(w\) for which we get an algorithm as in \Cref{thm:grid-minor-thm} is what we set as \(g(k,\max\{|P| \mid P \in \mathbb{P}\})\).
    Then, in time \(f(s)|V(G)|\) we can correctly output that \(\tw(G) \leq g(k)\), or a topological embedding \(h : H_s \to G\).
    In the former case there is nothing left to do, and in the latter case, by our choice of \(s\) using \Cref{thm:flat-grid}, we can find a subgrid
    \[H^* := H_{(k+1)\max\{|P| \mid P \in \mathbb{P} + 2\max\{|P| \mid P \in \mathbb{P}\}+3\}} \subseteq H_s\] such that the restriction \(h\big|_{H^*}\) of \(h\) to \(H^*\) is partially predrawn flat unless \(\ppdcrn(G,\Gamma) > k\).
    This allows us to correctly return \(\ftpcrn{\mathbb{P}}(G,\Gamma) > \ppdcrn(G,\Gamma) > k\) if no such \(H^*\) is found.

    It remains to consider the case in which we find \(H^*\).
    If there is an edge between endpoints with different colors in \(h(H^*_i)^+\) where \(H^*_i\) consists of the \((2\max\{|P| \mid P \in \mathbb{P}\}+2\})\)-innermost concentric cycles of \(H_i\), we fix this edge as \(e\).
    Otherwise all vertices of \(h(H^*_i)^+\) have the same color and we choose an arbitrary edge on the innermost concentric cycle of \(H^*\) as \(e\).
    
    We will show that \(\ftpcrn{\mathbb{P}}(G,\Gamma) \leq k\) if and only if \(\ftpcrn{\mathbb{P}}(G-e,\Gamma-e) \leq k\).

    By \Cref{lem:wlog-uncrossed} if \(\ftpcrn{\mathbb{P}}(G,\Gamma) \leq k\), then there is a drawing \(\mathcal{G}\) witnessing this in which all edges of \(h(H^*_i)^+\) are uncrossed.
    The removal of \(e\) from \(\mathcal{G}\) cannot create a crossing or inconsistencies with the remaining predrawing.
    Moreover, since \(e\) is uncrossed in \(\mathcal{G}\), the removal of its drawing cannot cause the non-existence of a crossing vertex in \(\mathcal{G}^\times\) which broke a pattern in \(\mathbb{P}\).
    By the definition of patterns, this is the only way in which a pattern in \(\mathbb{P}\) could occur in \(\mathcal{G} - e\) but not \(\mathcal{G}\).
    Because \(\mathcal{G}\) is \(\mathbb{P}\)-free, so is \(\mathcal{G} - e\) thereby witnessing \(\ftpcrn{\mathbb{P}}(G-e,\Gamma-e) \leq k\).
    
    Conversely, consider a drawing \(\mathcal{G}^-\) witnessing \(\ftpcrn{\mathbb{P}}(G-e,\Gamma-e) \leq k\).
    After the removal of \(e\), \(h(H_s) - e\) still contains a partially predrawn flat topological embedding of a grid with radius \(s - 1\).
    By \Cref{lem:wlog-uncrossed} we can assume all edges of \(h(H^*_i)^+ - e\) to be uncrossed in \(\mathcal{G}^-\).
    We can use the drawing of the \((\max\{|P| \mid P \in \mathbb{P}\} + 1)\)-outermost concentric cycle of \(h(H^*_i)^+ - e\) to glue in a crossing-free drawing of all vertices and edges that are drawn inside that cycle by \(\mathcal{G}^-\) into which \(e\) can be added without creating crossings, consistent with \(\Gamma\) using the fact that \(H^*\) is partially predrawn flat and \(e \in E(h(H^*)^+)\).
    We call \(\mathcal{G}\) the resulting drawing of \(G\).
    By construction, the number of crossings in \(\mathcal{G}\) is at most the number of crossings in \(\mathcal{G}^-\).
    
    We will show that any \(P \in \mathbb{P}\) occurring in \(\mathcal{G}\) had to already exist in \(\mathcal{G}^-\).
    First, note that \(e\) has distance at least \(\max\{|P| \mid P \in \mathbb{P}\} + 1\) from any crossing vertex in \(\mathcal{G}^\times\).
    By the crossing locality assumption this means the only way that \(e\) can be involved in the occurrence of a pattern \(P \in \mathbb{P}\) in \(\mathcal{G}\) is on a monochromatic path corresponding to an edge in \(E_P\) with both endpoints outside of \(h(H^*_i)^+\).
    If \(e\)'s endpoints are not of the same color then this is already a contradiction.
    Otherwise, by the choice of \(e\) all vertices of \(h(H^*_i)^+\) have the same color.
    Each pattern \(P \in \mathbb{P}\) has at most \(\mathcal{O}(\max\{|P| \mid P \in \mathbb{P}\})\) edges in \(E_P\) (recall that patterns are planar).
    If the path isomorphic to an edge in \(E_P\) for a pattern \(P \in \mathbb{P}\) occurring in \(\mathcal{G}\) contains \(e\), we can use the rest of \(h(H^*_i)^+\) to reroute the path.
    If this causes an intersection with a different such path used for the occurrence of the pattern, we iteratively keep rerouting along rows or columns which are closer and closer to the boundary of the grid. %
    Because we have \(2\max\{|P| \mid P \in \mathbb{P}\}+2\) columns and rows available, we start in the innermost concentric cycle (by the choice of \(e\)) we ultimately find the same pattern occurring in \(\mathcal{G}^-\).

    Because \(\mathcal{G}^-\) must be \(\mathbb{P}\)-free, the above argument implies the same for \(\mathcal{G}\).
    This means \(\mathcal{G}\) witnesses \(\ftpcrn{\mathbb{P}}(G,\Gamma) \leq k\) as desired.
\end{proof}

The above can be iteratively applied to conclude \(\ftpcrn{\mathbb{P}}(G,\Gamma) > k\) or delete edges until our instance has treewidth at most some \(g(k)\), achieving the goal of this section.

\section{Solving the Bounded Treewidth Case}
\label{sec:algorithm}
We now turn to the second part of proving our algorithmic result, which is to encode the fact that a graph has \(\ftpcrn{\mathbb{P}}\) at most \(k\) by an MSO-formula \(\varphi\) whose length is also bounded in our parameters over a graph that we obtain from \((G,\Gamma)\) without increasing its treewidth much.
Then we can use Courcelle's theorem~\cite{DBLP:journals/iandc/Courcelle90} to immediately show~\Cref{thm:ftpcr-fpt}.

First, we subdivide \(k\) times each edge of \(G\) and mark each of the subdivision vertices as \emph{crossing possibilities}.
It is well-known that edge-subdivision does not increase the treewidth by more than one if at all.
Whenever we are subdividing an edge in \(\Gamma\), we place its subdivisions at arbitrary distinct points on the drawing of that edge in \(\Gamma\) and consider them predrawn in this way.
Next, we use a minor of the construction used in \cite{DBLP:conf/compgeom/HammH22}, which allows us to also use their treewidth bound.
\begin{definition}[{cf.\ {\cite[Definition~3.2 of framings]{DBLP:conf/compgeom/HammH22}}; we differ only in the underlined part}]
    An \emph{orientation agnostic framing} of a partially predrawn graph \((G, \Gamma)\), is a (non-embedded) graph~$F$ constructed as follows.
    We start with the initial drawing \(\mathcal{D} = \Gamma\) and continue by the following steps in order:
	\begin{enumerate}
		\item While the underlying graph of \(\mathcal{D}\) is not connected, we iteratively add edges from \(G\) to \(\mathcal{D}\) that can be inserted in a planar way and which connect two previously disconnected components.
		If this is no longer possible while the graph is still disconnected, let \(B\) be a face of \(\mathcal{D}\) incident to more than one connected component.
		We pick a vertex \(v\) on \(B\) and connect \(v\) to an arbitrary vertex from each component incident to \(B\) which does not contain \(v\).
		We will call all edges added in this step the \emph{connector edges} (of the resulting framing).
		\item \label{step:edgesplit} We replace each edge $f=uw$ of \(\mathcal{D}\) from the previous step \underline{\smash{except the connector edges}} by three internally disjoint paths of length~$3$ between $u$ and~$w$.
		We will call these three paths together the {\em framing triplet of~$f$}, and denote by \(\mathcal{D}'\) the resulting drawing.
		\item \label{step:rotationscheme} Around each vertex $v\in V(\Gamma^\times)$ in \(\mathcal{D}'\) from the previous step, we add a cycle on the neighbors of $v$ in \(\mathcal{D}'\) in the cyclic order given by \(\mathcal{D}'\).
		We will call these cycles the \emph{framing cycles}, and all edges of the resulting planar drawing \(\mathcal{D}''\) the {\em frame edges}.
		\item \label{step:unite} Finally, we set \(F:=D''\cup G\) where \(D''\) is the underlying graph of \(\mathcal{D}''\) from the previous step.
  \end{enumerate}
\end{definition}
The only difference between the framings from \cite{DBLP:conf/compgeom/HammH22} and our \emph{(orientation agnostic) framings}, is that the former also adds framing triplets for connector edges.
The reason why we omit these edges is that we do not care about the orientation of disconnected parts of drawings with respect to each other. 
We can use the known fact that a framing of a graph does not have much larger treewidth than that graph~\cite[Lemma~7.11]{Hamm2022} and because of the minor-monotonicity of treewidth, the same treewidth bound also holds for orientation agnostic framings.

\begin{lemma}[{\cite[Lemma~7.11]{Hamm2022}}\protect\footnotemark]
\footnotetext{We cite the PhD thesis of one of the authors of \cite{DBLP:conf/compgeom/HammH22} because in the conference version the proof contains small mistakes.}
\label{lem:framing-bd-tw}
    Let \(F\) be a framing of a partially predrawn graph \((G, \Gamma)\).
    Then \(\operatorname{tw}(F) \leq f(\operatorname{tw}(G), k)\), where \(k = \ppdcrn(G,\Gamma)\) and \(f\) is a computable function.
\end{lemma}
We also use the MSO-encoding from \cite[Proof of Lemma~3.7]{DBLP:conf/compgeom/HammH22} to express that after the identification of at most \(k\) pairs of crossing possibilities, none of the obstructions to partially predrawn planarity occur. The encoding's length is bounded by a function of \(k\).

It remains to encode \(\mathbb{P}\)-freeness.
It is well-known to be MSO-encodable that a vertex set corresponds to an abstract (i.e.\ non embedded) small graph minor in the host graph.
However, we need to allow the abstract graph of a realization of a pattern to appear as a differently embedded graph.

To formalize this, let us reconsider in which ways embeddings play a role for the occurrence of patterns.
Recall that we constrain the maximum vertex degree to be at most \(2\) in the partially predrawn part of the patterns.
Hence, the only embedding-relevant information is which predrawn pattern cycles separate which other predrawn parts of the pattern from each other.
Because we require each face of the partial predrawing of a pattern to contain a planarized crossing or the endpoint of an edge which crosses, we only need to speak about the rotations around components with such vertices.

\subsection{Encoding Relevant Embedding Choices at 2- and 1-Vertex Cuts}
In a prospective planarized drawing, the only choices of embeddings we have are (Whitney) flips at original vertices which form cuts of size one or two in the planarized graph.
Such vertices have to already have been cuts of the same size in the original (non-planarized) graph.
We will slightly extend the graph we will apply Courcelle's theorem to, to speak about some of these choices in our MSO-formula.

Note, that we cannot hope to encode all choices without connecting all edges densely to each other, as otherwise this would allow us to encode the choice of a total order of arbitrary many pairwise independent elements via the choice of flips at a high degree vertex which is the center of a star.
This is impossible with a polynomial-length MSO-formula because, even using a small number of labels, these leaves are pairwise indistinguishable from each other.
Instead, we will only distinguish flips based on how they can flip around prospective crossings which are introduced by identifying pairs of crossing possibilities.
\Cref{fig:2-cut} gives an illustration to refer to for what we aim to and not to distinguish, as well as for later definitions.

\begin{figure}
    \centering
    \includegraphics[page=2]{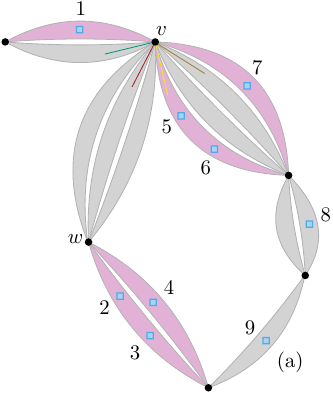}
    \hfil
    \includegraphics[page=3]{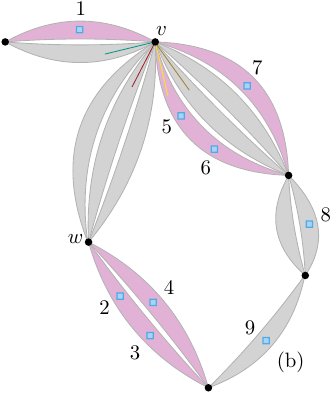}
    \includegraphics[page=4]{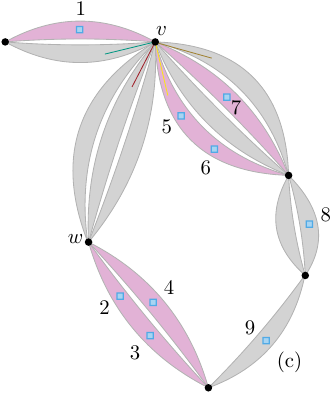}
    \hfil
    \includegraphics[page=5]{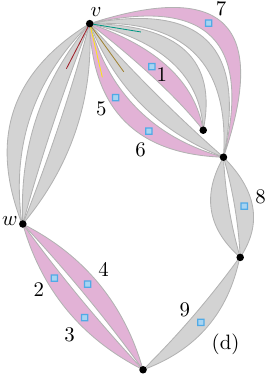}
    \caption{Illustration for embedding choices distinguished by the selection of rotation markers.
    Black disks are 1- and 2-vertex cuts which separate 3-connected components of a prospective planarization.
    Special vertices in 3-connected components are blue and indexed from \(1\) through \(9\).
    \(w\) is the only important vertex for \(v\).
    Components comprising important groups for \(v\) are highlighted.
    To encode (a), all selected rotation markers for edges incident to \(v\) have \(\{\{1\}, \{2,3,4\}, \{5,6\},\{7\}\}\) as first label entry and second label entry would specify the permutation in the written order.
    The exemplary colored edges incident to \(v\) would have the following third label entries in (a): green \(\rightarrow\) 1, red \(\rightarrow\) 1, yellow \(\rightarrow\) 0, brown \(\rightarrow\) 3.
    Here, the third entries of the labels of the selected rotation markers cannot allow us to distinguish the order of the three components between \(v\) and \(w\); they will be 1 for all edges leading into these components.
    Similarly, we cannot distinguish (a) and (b).
    However, our encoding distinguishes (a) and (b) from (c) by selecting a rotation marker with third label entry 4 instead of 3 for the brown edge.
    Encoding (d) corresponds to selecting rotation markers with the same first label entries as for (a), (b) and (c) but different second label entries.
    }
    \label{fig:2-cut}
\end{figure}

We can preprocess the original graph in polynomial time to find vertex cuts of size at most two.
For each edge of \(G\) that is incident to a vertex in such a cut, we attach the following vertices to its endpoints in the framing of \(G\).
We attach one distinct new vertex for 
and with each label of the form \((R, \pi, i)\) where \(R\) is a set of at most \(18k\) subsets of \([6k]\), \(\pi\) is a cyclic permutation on these sets, and \(i\) is the index of a position between at most two consecutive elements of \(\pi\).
We also attach two additional new vertices with labels \textit{first} and \textit{last} respectively to each edge incident to an at-most-two-sized-cut-set vertex.
For semantic reasons that will become clear later, we refer to the new vertices with labels of the form \((R,\pi,i)\) as \emph{rotation markers}, whereas we call the ones with labels \textit{first} or \textit{last} \emph{delimiter flags}.
We denote the resulting \emph{crossing-flip-aware framing} of \((G,\Gamma)\) by \(\overleftrightarrow{F}\).

\begin{restatable}{proposition}{flipframe}
\label{prop:flipframe}
    Let \(\overleftrightarrow{F}\) be a crossing-flip-aware framing of a partially predrawn graph \((G, \Gamma)\).
    The number of labels in \(\overleftrightarrow{F}\) is bounded in \(k\) and \(\operatorname{tw}(\overleftrightarrow{F}) \leq f(\operatorname{tw}(G), k)\), where \(k = \ftpcrn{\mathbb{P}}(G,\Gamma)\) and \(f\) is a computable function and \(\mathbb{P}\) is an arbitrary set of forbidden topological crossing patterns.
\end{restatable}
\begin{proof}
Note that the number of labels used for newly attached vertices is \((6k)^{18k} \cdot (18k)! \cdot (18k+1) \in k^{\mathcal{O}(k)}\) and the number of labels in the framing was constant.

Starting from a minimum-width tree decomposition of the framing, add for each new vertex attached to an edge a bag below an arbitrary bag containing both endpoints of that edge and add each new vertex into each one of these bags.
This is a tree decomposition for \(\overleftrightarrow{F}\).
Hence, the treewidth of \(\overleftrightarrow{F}\) is at most one larger than that of \(F\).
Now, the desired bound is a direct consequence of \Cref{lem:framing-bd-tw}, and that by definition \(\ppdcrn(G,\Gamma) \leq \ftpcrn{\mathbb{P}}(G,\Gamma)\).
\end{proof}

In an MSO-formula we can use an existentially quantified vertex set variable \(M\) to select precisely one rotation marker for each of the edges for which they were added, and precisely one marker with label \textit{first} and one with label \textit{last} for the set of edges leading into each 3-connected component.
We will describe the meaning of such a selection for the choice of embedding for a planarization and how we encode it in our MSO-formula \(\varphi\) after indicating the structure of the MSO-formula which we aim for.

The structure of the MSO-formula which we target is as follows where the outermost (i.e.\ the ones that are explicit in the below structure) existentially quantified variables \(x_1, \dotsc, x_k, y_1, \dotsc, y_k\) fix the choices of identified crossing possibilities and \(M\) fixes the choices of markers and delimiter flags as described above.
\begin{align*}
    \varphi := & \exists x_1 \exists y_1 \dotsc \exists x_k \exists y_k \exists M \ \begin{aligned}[t] & \textsf{PartiallyPredrawnPlanarAfterIdentifying}(x_1,y_1, \dotsc, x_k,y_k)\\
    \land \ & \textsf{ValidMarkerFlagChoice}(M,x_1,y_1, \dotsc, x_k,y_k)
    \end{aligned}\\
    & \quad \bigwedge_{(P,\Phi) \in \mathbb{P}} \neg \big(
    \begin{aligned}[t] & \exists z_1 \dots \exists z_{|P|} \textsf{FormAbstractPInPlanarization}(z_1,\dotsc,z_{|P|},x_1,y_1, \dotsc, x_k,y_k)\\
    & \quad \land \textsf{Like\(\Pi\)InEmbeddedPlanarization}(z_1, \dotsc, z_{|P|},M,x_1,y_1, \dotsc, x_k,y_k)\big),\end{aligned}
\end{align*}
\begin{description}
\item[\textsf{PartiallyPredrawnPlanarAfterIdentifying}] denotes the aforementioned formula from the proof of \cite[Lemma~3.7]{DBLP:conf/compgeom/HammH22} that expresses that identifying the crossing possibilities given by the existentially quantified \(x_i,y_i\) indeed results in a drawing extension in which no further crossings are necessary, and 
\item[\textsf{FormAbstractPInPlanarization}] denotes the aforementioned formula to find an occurrence of the abstract graph of a pattern where \(z_1, \dotsc, z_{|P|}\) take the roles of the vertices of this abstract pattern graph.

\item[\textsf{ValidMarkerFlagChoice}] should express a valid selection of rotation markers which allows us to assume specific embedding behavior relative to important vertices in the planarization, specifically planarized crossings and endpoints of crossing edges:
Fix an indexing of all identified crossing possibilities---\(x_1, y_1, \dotsc, x_k, y_k\)---and endpoints in \(V(G)\) of edges in \(E(G)\) turned paths in \(\overleftrightarrow{F}\) on which identified crossing possibilities lie. 
We can access these endpoints using existentially quantified vertex variables and assert in MSO that they are indeed endpoints of a path that corresponds to an edge in \(E(G)\) on which a crossing possibility from the first existentially quantified vertex pair lies.
We call these indexed vertices \emph{special}, and note that there are at most \(6k\) of them (at most \(2k\) identified crossing possibilities and at most \(4k\) endpoints of prospective crossing edges).

The set of all 1- and 2-vertex cuts in the prospective planarization\footnote{In the planarization resulting from the identification of some existentially quantified vertices as prospective crossing it may be that a 1- or 2-vertex cut in the original graph is no longer one in the planarization, so we cannot simply use the set of vertices from the preprocessing step in which we introduced triangles.} can be described straightforwardly via a vertex set variable in MSO which also allows us to speak of the connected components that remain after the removal of all vertices in such cuts.
We call a 3-connected component \emph{important} if it contains a special vertex.
Let \(v \in V(G)\) be in a 1 or 2-vertex cut in the prospective planarization.
We say a vertex \(u\) of \(G\) other than \(v\) in a 1 or 2-vertex cut in the prospective planarization is \emph{important for \(v\)} if there is no path from \(v\) to \(u\)
using a minimum number of 3-connected components and at least one special vertex
from a 3-connected component in the prospective planarization.\footnote{From now on we will usually talk about connectivitiy relationships in the prospective planarization, without explicitly mentioning this every time.}
We call a set of special vertices an \emph{important group for \(v\)} if it is the set of special vertices in a 3-connected component incident to \(v\) or the union of all important vertices in 3-connected components incident to a vertex that is important for \(v\).

By definition, there are at most \(18k\) important groups for \(v\): the important components incident to \(v\) can at most partition \(6k\), and each special vertex can be in 3-connected components incident to at most two important vertices for \(v\), and hence in at most two groups which yields an upper bound of \(12k\) on groups incident to important vertices.
The label \((R,\pi,i)\) of the unique rotation marker with such a label selected for an edge incident to one vertex \(v\) should encode:
(i) in its first entry, the important groups for \(v\); 
(ii) in its second entry, in which rotation order the important groups from the first entry are arranged around \(v\);
and  (iii) in its third entry, an index for which position between these groups the edge lies in the rotation order around \(v\);
    dummy value 0 is used when the edge is from \(v\) to an important component incident to \(v\).

Further, the selection of delimiter flags with label \textit{first} (respectively \textit{last}) for an edge from \(v\) into an important component incident to \(v\) should encode that the edge of which the \textit{first}-labeled (respectively \textit{last}-labeled) vertex is selected is the first (respectively last) to its non-\(v\) endpoint's 3-connected component in the rotation around \(v\).

Describing these desired semantics is possible in MSO as follows.
We encode in MSO that the first entry of the labels of the rotation markers selected for edges incident to one vertex \(v\) must be equal to the important groups for \(v\), which is possible since the definition of groups speaks about the existence of certain paths.
In this way, the first entries of the labels of all selected rotation markers are already uniquely determined after identifying prospective crossings and we merely include them in the labels for convenience to ensure that the available second and third entries are compatible with the correct partition.

The second and third entries of the labels of selected rotation markers encode actual embedding choices and we need to make sure that these embedding choices are made consistently.
This is straightforward for edges incident to a common vertex \(v\):
Since the first and second entries are constrained to be equal for labels of rotation markers selected for such edges, we merely need to disallow the existence of paths from \(v\) to two special vertices in important groups for \(v\) such that the first edge on the path to the special vertex which is first in the cyclic order according to the second label entry is at a strictly larger position according to its selected rotation marker's third label entry than the first edge of the path to the special vertex which is later in the second label entry.
We do not express the comparisons of positions in MSO (which is not part of its expressive power) but instead hard-code this by a disjunction on the conjunction of pairs one of \(f(k)\) labels which results in a formula \textsf{LocallyValidMarkerFlagChoice} whose length is bounded in \(k\).

The only other way in which inconsistencies can occur in the selection of rotation markers is when rotation markers at edges incident to different vertices encode partial rotations that contradict each other.
By construction, this means that there is a cut consisting of two vertices in the prospective planarization and two disjoint paths between these vertices, such that the labels of the rotation markers selected for the first and last edges of these paths describe contradictory positions in the rotations around the two considered cut vertices.
Again, the existence of such paths can be expressed (and hence also forbidden) in MSO with the contradictions between labels encoded via explicit \(f(k)\)-length formulas.
See \Cref{fig:Tconsistent}.

\begin{figure}
    \centering
    \includegraphics{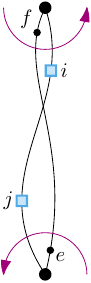}
    \hfil
    \includegraphics[page=2]{figs/incompatible-T-choice.pdf}
    \hfil
    \includegraphics[page=3]{figs/incompatible-T-choice.pdf}
    \caption{Depiction of how incompatible \(M\) choices can look like. In each case the incompatibility must be detectable around the vertices which are the cut vertices around which incompatible vertices where selected in \(M\).
    We can check for the existence of identified crossing possibility or endpoint of crossing edges indices \(i\) and \(j\) which -- by the labels of the inconsistently chosen \(M\) vertices -- force inconsistent embeddings of the paths starting in edge \(e\) and \(f\) (and potentially ending in \(f'\) and \(e'\) respectively).
    Such incompatibilities can be explicitly listed as combinations not allowed to occur among the labels of consistently selected \(M\) vertices.
    This is somewhat unwieldy but possible via \(f(k)\)-length MSO-formulas.
    As an example, situations like the left one can be precluded by forbidding the existence of paths between two vertices of a 2-vertex cut such that one path contains special vertices indexed \(i\) and \(j\) (possibly \(i=j\) where \(i\) is in an important group for one considered cut vertex and \(j\) is for the other and the choice of rotation markers encodes that the first and last edge on the other paths are both before (or both after) the important groups containing \(i\) and \(j\) in the rotations around the considered cut vertices.}
    \label{fig:Tconsistent}
\end{figure}
In the selection of delimiter flags, we can similarly explicitly prohibit contradictory choices of delimiter flags for the same component around different 1- or 2-vertex cut vertices.
Beyond this we need to ensure that the combination of \textit{fist} and \textit{last} delimiter marker are compatible for each component incident to an arbitrary 1- or 2-vertex cut vertex \(v\).
For this, we express in MSO that there must be a path in a 3-connected component under consideration that visits all neighbors of \(v\) in that component and forms a cycle together with the edges whose delimiter flags were selected or that there is only one edge from \(v\) to the component in question.
This MSO-expression can then be conjoined with \textsf{LocallyValidMarkerFlagChoice} to form \textsf{ValidMarkerFlagChoice}.

\item[\textsf{Like\(\Pi\)InEmbeddedPlanarization}] should express that \(z_1, \dotsc, z_{|P|}\), which we know occur in a way that realizes the abstract structure of \(P\) if \(\textsf{FormAbstractPInPlanarization} (\)\(z_1, \dotsc, z_{|P|}\), 
\( x_1, y_1, \dotsc, x_k, y_k)\), also occur in a way that realizes the predrawing \(\Pi\) of \(P\).
We will use the predicate described by the statement of the following lemma.%
\begin{lemma}
\label{lem:3consepmso}
    There is a constant-length MSO-formula \(\textsf{3ConSep}(S,a,b)\) with a free edge set variable \(S\) and two free vertex variables \(a\) and \(b\) that evaluates to true on an arbitrary 3-connected graph \(G\) if and only if the drawings of \(a\) and \(b\) are separated from each other by the drawing of \(S\) in the unique embedding of \(G\).
\end{lemma}

Before we argue the correctness of this lemma, let us formalize what we mean by two vertices being separated by an edge set \(S\) (which does not necessarily correspond to a closed simple curve).
In particular, for our purposes when applying \Cref{lem:3consepmso}, we will always be in the situation in which \(S\) is the intersection of a cycle in a supergraph of \(G\) in which \(G\) is a 3-connected component intersected with \(G\).
This means \(S\) will be the edge set of a collection of paths whose endpoints are on one face cycle of \(G\).
We can connect the endpoints of each of these paths in a topologically unique way that the resulting closed simple curves are pairwise non-nesting and the completing curves do not intersect the drawing of \(G\).
We speak about topological separation of \(a\) and \(b\) with respect to these closed simple curves which we denote by \(\mathcal{S}\).
The following lemma forms the basis of the MSO-encoding for \Cref{lem:3consepmso}.
\begin{lemma}
\label{lem:3consep}
Let \(G\) be a 3-connected graph which we consider embedded in the unique way on a sphere and \(a,b \in V(G)\).
Further, let \(\mathcal{S}\) be the drawing of a cycle in \(G\) or a set of pairwise non-nesting simple closed curves that each consist of the drawing of a path in \(G\) which is closed with a curve that does not intersect \(G\) such that all closing curves lie on the same face of \(G\).
    If \(a\) and \(b\) are in different connected components of \(G - V(S)\), they lie on the same side of \(\mathcal{S}\) if and only if there is a sequence of edges in \(E(G) \setminus S\) starting with one that shares a face with \(a\) and ending with one that shares a face with \(b\) such that precisely the pairwise consecutive elements in this sequence share a face.
\end{lemma}
\begin{proof}
    For the forward direction of the proof, consider \(a\) and \(b\) as being on the same side of \(\mathcal{S}\).
    Note that since the embedding of \(G\) is unique, so is the dual.
    Further, because they are on the same side of \(\mathcal{S}\) there must be a path between a face containing \(a\) and a face containing \(b\) in the dual of the embedding of \(G\) that does not use any edges corresponding to edges in \(S\).
    This path can be translated into a sequence of edges of \(E(G) \setminus S\) starting with one that shares a face with \(a\) and ending with one that shares a face with \(b\) where pairwise consecutive subdivision vertices share a face.
    This concludes the proof of the forward direction.

    For the backward direction, we proceed by contraposition, i.e.\ consider \(a\) and \(b\) being on different sides of \(S\).
    This means that any simple curve avoiding all drawings of vertices and connecting \(a\) and \(b\) must transversely intersect \(\mathcal{S}\).
    At the same time the existence of a sequence of edges as described in the lemma statement would imply the existence of precisely such a curve which does not intersect \(\mathcal{S}\), so it is precluded. 
\end{proof}

With the above lemma in hand, we can turn to the proof of \Cref{lem:3consepmso}.
\begin{proof}[Proof of \Cref{lem:3consepmso}]
The fact that \(a\) and \(b\) are in different connected components of \(G - V(S)\) can easily be encoded in MSO.
According to \Cref{lem:3consep}, we are left with encoding the existence of a sequence of edges in \(E(G) \setminus S\) starting with one that shares a face with \(a\) and ending with one that shares a face with \(b\) such that precisely the pairwise consecutive elements in this sequence share a face.
There are two points we need to address for this.

Firstly, MSO is not aware of faces or what it means to share a face.
This difficulty can be overcome by the well-known fact that faces of embeddings of 3-connected graphs are uniquely determined as being induced cycles whose removal does not disconnect the graph~\cite{DBLP:books/daglib/0030488}.
Two objects being together on some such cycle, i.e.\ sharing a face, can be easily expressed in MSO.

Secondly, MSO cannot directly speak of sequences of arbitrary length in a constant length formula.
However, we can use a standard trick (which is also used to encode connectivity without being able to speak of paths of arbitrary length explicitly).
We express the existence of an unordered set of edges in which for every non-trivial proper subset of that set an edge shares a face with some edge in the complement of that subset, all but two edges in the set share an edge with precisely two other edges in the set, and the two edges that do not share a face with exactly one other edge in the set each and share a face with \(a\) and \(b\) respectively.

Overall this shows how to formulate the \textsf{3ConSep} predicate described in \Cref{lem:3consepmso}.
\end{proof}

\begin{figure}
    \centering
    \includegraphics[page=2]{figs/markedpaths.pdf}
    \hfil
    \includegraphics{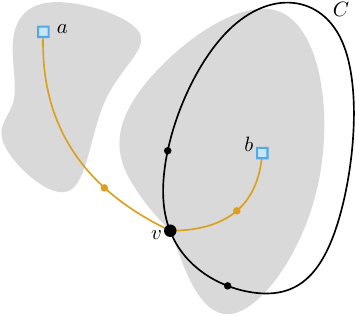}
    \caption{Illustrations of cases in which there is a 1- or 2-vertex cut including \(v \in V(C)\).
		If \(a\) and \(b\) lie in different 3-connected components than the neighbors of \(v\) (left) then the labels of rotation markers allow us to check whether our selection of rotation markers encodes that the edges from \(v\) to its neighbors on \(C\) alternate with the ordering of the important groups for \(v\) that are between \(v\) and \(a\) and \(v\) and \(b\) respectively.
        In cases where \(a\) or \(b\) (or both) lie in the same 3-connected component as one or both of the neighbors of \(v\) on \(C\) (e.g.\ right), the order of edges into the 3-connected components containing \(a\) or \(b\) is fixed by their selected \textit{first} and \textit{last} delimiter flags in a way that is encodable in MSO: specifically, two edges into one of these components from \(v\) are in the same order as they appear on a path from the non-\(v\) endpoint of the edge of which the \textit{first}-labeled marked was selected to the non-\(v\) endpoint of the edge of which the \textit{last}-labeled marked was selected that contains all other non-\(v\) endpoints of edges from \(v\) into that component.
		In combination with the selected rotation markers this allows us to express that our selection of \(M\) results in the edges from \(v\) to its neighbors on \(C\) alternating with the edges from \(v\) to important groups for \(a\) between \(v\) and \(a\) and \(v\) and \(b\) respectively.}
    \label{fig:separating}
\end{figure}
\begin{figure}
    \centering
    \includegraphics[page=4]{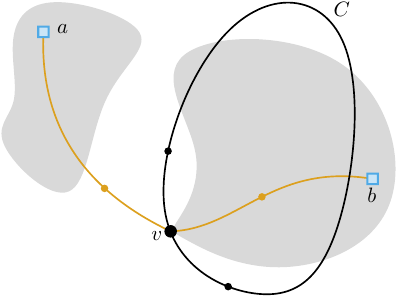}
    \hfil
    \includegraphics[page=3]{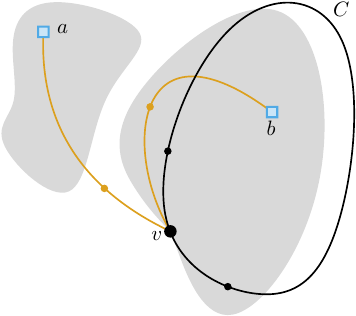}
    \caption{Left: Case in which an alternating rotation between connections from \(v\) to \(a\) and \(b\) with neighbors on \(C\) does not imply that \(C\) topologically separates \(a\) and \(b\).
    In this case we additionally need to check that the first vertex on the path to \(b\) which is between the edges of \(C\) incident to \(v\) is not topologially separated from \(b\) by edges in \(C\) in the 3-connected component containing \(b\) using \textsf{3ConSep}.
    Right: Case in which \(C\) topologically separates \(a\) and \(b\) but the first edge of a \(v\)-\(b\)-path is not between the edges of \(C\) incident to \(v\) in the rotation around \(v\).
    In this case we need to check whether the non-\(v\) endpoint of the first edge of such a path is topologically separated from \(b\) by edges in \(C\) in the 3-connected component containing \(b\) using \textsf{3ConSep}.}
    \label{fig:combiningbetweensep}
\end{figure}

\textsf{Like\(\Pi\)InEmbeddedPlanarization} can be composed of (negations of) building blocks that a cycle \(C\) in the prospective planarization given by the cyclic sequence of its vertices (which we have access to as vertex variables) topologically separates two vertices \(a\) and \(b\) (also given as vertex variables) in any drawing conforming to the existentially quantified choices of identified crossing possibilities and rotation markers.
These blocks and their negations can be composed in a way that is specific to \((P,\Pi)\).
Note that, because we require each connected component of the predrawn subgraph of a forbidden pattern to contain a planarized crossing or the endpoint of a crossing edge, we can always choose \(a\) and \(b\) to be such vertices.

\begin{proposition}
\label{prop:3Csep}
    If all of \(C\) is in a 3-connected component, there is a constant length MSO-formula expressing that \(C\) topologically separates \(a\) and \(b\).
\end{proposition}
\begin{proof}
If all of \(C\) is in a 3-connected component in which neither of \(a\) and \(b\) is, then it is obvious that the drawing of \(C\) cannot separate the drawings of \(a\) and \(b\), so we can explicitly exclude this behavior via a negated MSO-formula when formulating our building blocks.
If all of \(C\) is in a 3-connected component in together both \(a\) and \(b\), we can trivially invoke \(\textsf{3ConSep}(C,a,b)\) 
 from \cref{lem:3consepmso}.
Lastly, if all of \(C\) is in a 3-connected component with exactly one of \(a\) and \(b\) (without loss of generality \(a\)) we can invoke a disjunction of \(\textsf{3ConSep}(C,a,x)\) over all \(x\) in the 3-connected component of \(C\) that lie on an \(a\)-\(b\)-path in the planarization and the following:
The 3-connected component without \(C\) is connected (in this case it is a face cycle in the unique embedding of that 3-connected component) and \(C\) vertex-separates \(b\) from each vertex in the 3-connected component (then it is a face cycle separating \(a\) and \(b\) in a topological sense.
\end{proof}

We are left with the case in which there is a 1- or 2-vertex cut vertex \(v\) on \(C\).
Our building block formulas can check conditions involving \(M\), that encode that connections from \(v\) to \(a\) and from \(v\) to \(b\) alternate with the edges of \(C\) in the rotation around \(v\).
See \Cref{fig:separating}.
These conditions alone do not yet imply that the drawing of \(C\) separates those of \(a\) and \(b\).
However, they can be combined with \textsf{3ConSep} predicates from \Cref{lem:3consepmso} to express exactly this.
See \Cref{fig:combiningbetweensep}.
\end{description}

Overall, this implies the encodability of the existence of a drawing extension of \((G,\Gamma)\) with at most \(k\) crossings without a realization of a forbidden pattern in \(\mathbb{P}\) via an MSO-formula on \(\overleftrightarrow{F}\) whose length is bounded in \(k\), \(|\mathbb{P}|\) and \(\max_{P \in \mathbb{P}} |P|\).
\begin{lemma}
\label{lem:totalmsoencoding}
    Given a set of topological crossing patterns \(\mathbb{P}\), there is an MSO-formula \(\varphi\) whose length is bounded in \(k\), \(|\mathbb{P}|\) and \(\max_{P \in \mathbb{P}} |P|\) such that \(\varphi\) can be satisfied on a crossing-flip-aware framing \(\overleftrightarrow{F}\) of any partially predrawn graph \((G,\Gamma)\) if and only if \(\ftpcrn{\mathbb{P}}(G,\Gamma) \leq k\).
\end{lemma}
\begin{proof}
We construct \(\varphi\) as described in the course of this subsection.
The desired bound on its length in \(k\), \(|\mathbb{P}|\) and \(\max_{P \in \mathbb{P}} |P|\) holds by construction and the length bounds on the utilized building blocks described in \Cref{lem:3consepmso} and \Cref{prop:3Csep}.

Fix a crossing-flip-aware framing \(\overleftrightarrow{F}\) of an arbitrary partially predrawn graph \((G,\Gamma)\).

It remains to show that \(\varphi\) is satisfiable on \(\overleftrightarrow{F}\) if and only if there is a drawing of \(G\) witnessing \(\ftpcrn{\mathbb{P}}(G,\Gamma) \leq k\).

For the forward direction, assume there is a satisfying assignment of \(x_1, \dotsc x_k, y_1, \dotsc, y_k, M\) for \(\varphi\).
By definition of \(\varphi\) this means \(\textsf{PartiallyPredrawnPlanarAfterIdentifying}(x_1, y_1, \dotsc, x_k,y_k)\) holds.
\(\textsf{PartiallyPredrawnPlanarAfterIdentifying}(x_1, y_1, \dotsc, x_k,y_k)\), by the correctness the proof of \cite[Lemma~3.7]{DBLP:conf/compgeom/HammH22}, implies that there is a drawing of \(G\) witnessing that \(\ppdcrn(G,\Gamma) \leq k\).
We claim that there is also such a drawing which respects the semantics of the choices made by \(M\).

Assume for contradiction that this is not the case, i.e.\ conforming to the semantics of the choices made by \(M\) implies a crossing which is not planarized by the identification of the crossing possibility pairs \(x_1, y_1, \dotsc, x_k, y_k\).
We call such a crossing \emph{unaccounted}.
Because \(G\) is connected and an unaccounted crossing must be forced by \(M\), the edges involved in an unaccounted crossing must be connected by a path which has a vertex \(v\) in a 1- or 2-vertex cut of \(G\) in which the pairs of \(x\) and \(y\) vertices with the same index in \(x_1, y_1, \dotsc, x_k, y_k\) are identified.

By well-known characterizations of all possible embeddings of planar graphs, the only different
embedding choices of a planar graph come from consistent rotation permutations and mirroring of components at 1-vertex and 2-vertex cuts~\cite{MacLane,DBLP:reference/crc/Patrignani13}.
\(M\) partially specifies these permutations and mirrorings; the only way in which it can force unaccounted crossings is by inconsistent specifications.
The unaccounted crossing cannot be forced by \(M\)-choices for edges incident to a single such vertex \(v\), i.e. at a 1-vertex cut, as \(\varphi\) contains in a conjunction \textsf{ValidMarkerFlagChoice} which in turn contains in a conjunction \textsf{LocallyValidMarkerFlagChoice} which precludes precisely the combination of choices in \(M\) for edges incident to a single vertex resulting in unaccounted crossings.
At 2-vertex cuts, the edges involved in a hypothetical unaccounted crossing must be connected by paths which each have a vertex \(v\) in the 2-vertex cut of the prospective planarization in which \(x_1, y_1, \dotsc, x_k, y_k\) are identified.
However, this is precluded by the second part of \textsf{LocallyValidMarkerFlagChoice} which is contained in \(\varphi\) conjunctively.
This means we get a contradiction to an unaccounted crossing existing and there is a drawing of \(G\) that witnesses \(\ppdcrn(G,\Gamma) \leq k\) which respects the semantics of the choices made by \(M\).
In any\footnote{Notice that by construction, \(M\) fixes all embedding choices relative to planarized crossings and endpoints of crossing edges which are the only ones relevant for the existence of topological crossing patterns we target in our framework.} such a drawing the existence of a topological crossing pattern \(P \in \mathbb{P}\) is precluded by 
\begin{align*}& \neg \big(\exists z_1 \dots \exists z_{|P|} \textsf{FormAbstractPInPlanarization}(z_1,\dotsc,z_{|P|},x_1,y_1, \dotsc, x_k,y_k)\big)\\
    & \quad \land \textsf{Like\(\Pi\)InEmbeddedPlanarization}(z_1, \dotsc, z_{|P|},M,x_1,y_1, \dotsc, x_k,y_k)\big).\end{align*}
In particular, the intended semantics of \textsf{FormAbstractPInPlanarization} can be achieved by standard and folklore techniques and the intended semantics of \textsf{Like\(\Pi\)InEmbeddedPlanarization} holds by the specification of its building blocks from \Cref{lem:3consepmso} and \Cref{prop:3Csep}.
Thus such a drawing witnesses \(\ftpcrn{\mathbb{P}}(G,\Gamma) \leq k\) as desired.

For the backward direction assume a drawing witnessing \(\ftpcrn{\mathbb{P}}(G,\Gamma) \leq k\).
It is straightforward to verify that choosing \(x_1, y_1, \dotsc x_k, y_k\) as the subdivision vertices of edges that occur in a pair of the at most \(k\) crossings and \(M\) according to the rotation order around 1- and 2-vertex cut vertices in the planarization of the drawing and the orientation of its \(3\)-connected components results in an assignment of its existential variables which makes \(\varphi\) true in \(\overleftrightarrow{F}\).
\end{proof}

Together with the reduction to bounded treewidth from \Cref{sec:bd-tw} and Courcelle's theorem this establishes our main algorithmic result, \Cref{thm:ftpcr-fpt}.
\begin{proof}[Proof of \Cref{thm:ftpcr-fpt}]
Let \((G,\Gamma)\) be a partially predrawn graph and \(k \in \mathbb{N}\).
While \(\tw(G) > g(k,\max \{|P| \mid P \in \mathbb{P}\})\) where \(g\) is as in \Cref{lem:bd-tw}, we can apply \Cref{lem:bd-tw} to correctly determine that \(\ftpcrn{\mathbb{P}} > k\), or find an edge whose deletion from the partially predrawn graph does not change whether its \(\ftpcrn{\mathbb{P}}\) is at most \(k\).
This can be iterated at most \(|E(G)|\) times, and each iteration runs in \FPT-time with respect to \(k\) and \(\max \{|P| \mid P \in \mathbb{P}\}\) according to \Cref{lem:bd-tw}.
After having done this, we have correctly decided \(\ftpcrn{\mathbb{P}} > k\), or have reduced to an equivalent instance \((G,\Gamma)\) with \(\tw(G) \leq g(k,\max \{|P| \mid P \in \mathbb{P}\})\).
In the former case, there is nothing left to do.
In the latter, we can construct (by definition in polynomial time) the crossing-flip-aware framing \(\overleftrightarrow{F}\) of \((G,\Gamma)\).
By \Cref{prop:flipframe}, \(\tw(\overleftrightarrow{F}) \leq f(\tw(G),k)\) and by our obtained treewidth-bound for \(G\),  overall \(\tw(\overleftrightarrow{F}) \leq h(k, \max \{|P| \mid P \in \mathbb{P}\})\) for some computable function \(h\).
We can construct the MSO-formula \(\varphi\) from the statement of \Cref{lem:totalmsoencoding} as described in this subsection in \FPT-time parameterized by \(k\), \(|\mathbb{P}|\) and \(\max \{|P| \mid P \in \mathbb{P}\}\).
The length of \(\varphi\) is bounded in \(k\), \(|\mathbb{P}|\) and \(\max \{|P| \mid P \in \mathbb{P}\}\) by \Cref{lem:totalmsoencoding}.
Courcelle's theorem can be used to decide \(\varphi\) on \(\overleftrightarrow{F}\) in \FPT-time parameterized by \(|\varphi|\) and \(\tw(\overleftrightarrow{F})\), i.e.\ using our previously argued bounds \FPT-time parameterized by \(k\), \(|\mathbb{P}|\) and \(\max \{|P| \mid P \in \mathbb{P}\}\).
By \Cref{lem:totalmsoencoding}, this is equivalent to deciding \(\ftpcrn{\mathbb{P}}(G,\Gamma) \leq k\).
\end{proof}

\section{W[1]-Hardness of \textsc{Straight-Line Planarity Extension}}
\label{sec:hardness}

In this section, we show that the straight-line constraint makes the problem of just inserting a small number of edges much harder than for topological constraints, even without crossings. 

\begin{restatable}{theorem}{thmwonehard}
\label{thm:W1hard}
\textsc{Straight-Line Planarity Extension} is W[1]-hard. 
\end{restatable}

We present a parameterized reduction from the \W[1]-hard {\sc{Grid Tiling}} problem~\cite{DBLP:conf/focs/Marx07a}.

\textsc{Grid Tiling} takes as input a $k \times k$ grid in which each cell $(i,j)$ contains a $m\times m$ subgrid with some \emph{valid tiles} (given as subset $S_{i,j} \subseteq [m] \times [m]$), where \(m \in \mathbb{N}\) and \(k \in \mathbb{N}\) is the parameter.
The question is whether there is a way to pick one tile per cell so that in each row and column, the picked tiles have the same second and first coordinate, respectively.

\subsection{Reduction overview}
Given an instance of {\sc{Grid Tiling}} we construct a straight-line plane graph of linear size and polynomial vertex coordinates. 
This is the predrawn graph $\Gamma$ of the {\sc Straight-Line Planarity Extension}. 
We denote by $H$ the subgraph of the whole graph $G$ induced by the edges that are not part of $\Gamma$.
While in general, \(G = H + \Gamma\) is not true by definition in general, it will be the case in the instances produced by our reduction.
In our construction, $H$ is a graph of size $O(k^2)$. %

\begin{figure}
    \centering
    \includegraphics[page=2,width=\linewidth]{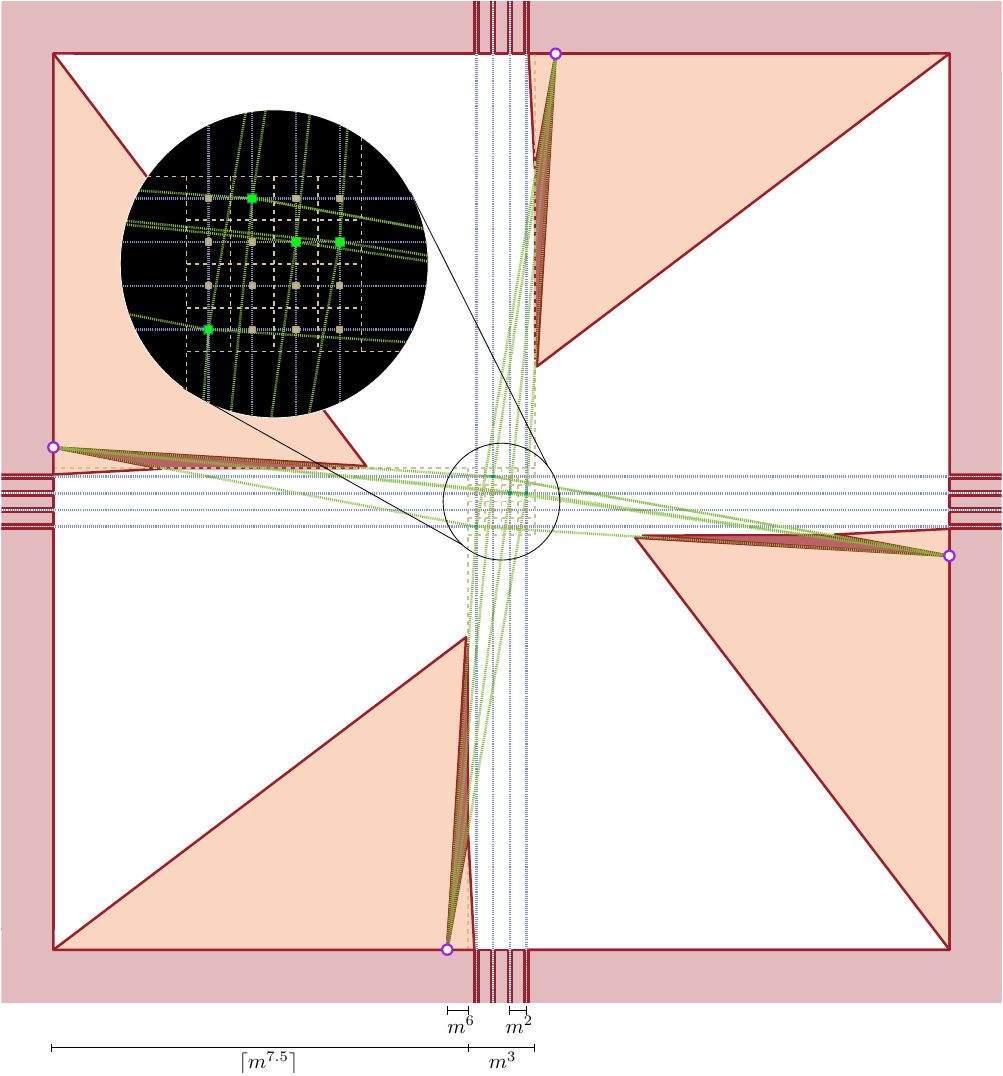}
    \caption{Illustration of the reduction (not up to scale). The shaded red and orange areas represent obstacles that are effectively blocked by predrawn edges. Dotted lines represent connections allowed by the obstacles, while dashed lines help gadget placement.
    The central dashed square separates the inner and middle layers of the cell. 
    The outer red obstacles lie in the outer layer.
    The only regions in a cell that admit non-blocked  connections to the four purple filtering vertices and straight-line connections to neighboring cells are in the green regions/dots in the core of the cell.}
    \label{fig:reduction-cell}
\end{figure}

The high-level idea is that
for each grid cell $(i,j) \in [k] \times [k]$
the predrawn graph $\Gamma$ defines certain \emph{core regions} that are in one-to-one correspondence with tiles; 
see the central small squares in \cref{fig:reduction-cell}. 
We show that in any solution to the partially predrawn rectilinear crossing number instance that we construct, certain vertices $\xi_{i,j}$ of $H$ must be drawn each in a region corresponding to a valid tile of cell $(i,j)$. 
This is done by including in $H$ four connections to vertices on the sides of the cell.
Building obstacles that constrain these connections to only allow to use valid core regions corresponding to valid tiles requires a very careful geometric construction and argumentation. 
An additional challenge is that we cannot block vertical and horizontal connections from $\xi_{i,j}$, as we need those to ensure that the $\xi_{i,j}$'s are drawn in core regions which are vertically and horizontally aligned across the grid.
With another set of predrawn \emph{obstacles} we can guarantee 
such alignment,
yielding a valid solution for the {\sc{Grid Tiling}} instance.
In the other direction, from a solution to the  {\sc{Grid Tiling}} instance we can easily construct a solution to the \textsc{Straight-Line Planarity Extension} instance by placing the vertices $\xi_{i,j}$ in the centers of the core regions that correspond to the \textsc{Grid Tiling} solution; see \cref{fig:reduction-solution}. 

Edges can be added to connect the obstacles, implying \W[1]-hardness of \textsc{Partially Predrawn Rectilinear Crossing Number}. %

\begin{figure}
    \centering
    \includegraphics[page=3,width=\linewidth]{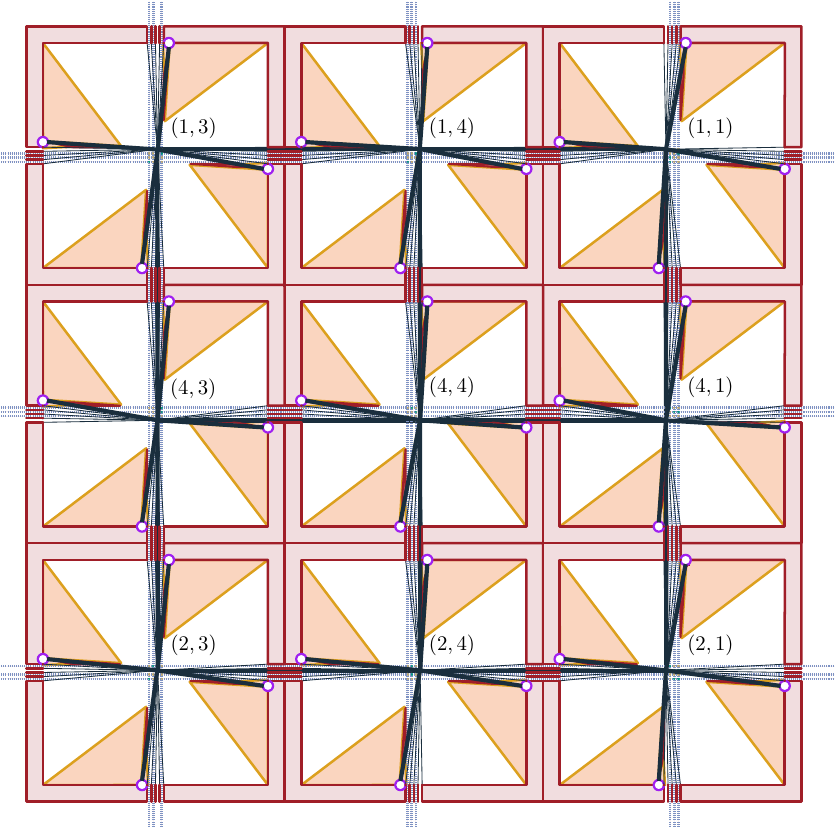}
    \caption{In black, a solution to the \textsc{Straight-Line Planarity Extension} instance produced by the reduction, corresponding to a {\sc{Grid Tiling}} solution.}
    \label{fig:reduction-solution}
\end{figure}

\subsection{Predrawn graph}
We first describe the predrawn graph $\Gamma$ constructed from a given instance of {\sc{Grid Tiling}}.

\paragraph*{Obstacles and channels}
We restrict certain areas of the plane using predrawn edges. 
In particular, any closed polygonal chain defines an \emph{obstacle}. 
In the following, we are going to use obstacles to form \emph{channels} through which edges can pass; see~\cref{fig:channel} (left). 
For two obstacles having two parallel boundary pieces that bound the long sides of an empty rectangle between them, we call an inclusion-maximal such rectangle a \emph{channel}. 
The \emph{length} and \emph{width} of the channel are the length of the longest and of the shortest side, respectively.
For two obstacles sharing a vertex and having two incident boundary pieces that bound the longest leg and the hypotenuse of an empty triangle between them, 
we also call an inclusion-maximal such a triangle a \emph{channel}. 
The \emph{length} and \emph{width} of the channel are the length of the longest and of the shortest leg, respectively.

Crucially, channels constrain where the edges passing through them can have their endpoints outside or on the vertices of the channel. 
Such edges must be contained in what we call the \emph{funnel of the channel}. 
Given a channel defined by a rectangle $R$ with distinct side lengths, 
consider the diagonals that pass through two opposite corners.
The two diagonals define two double wedges: one with larger angles and one with smaller ones. 
The {funnel of the channel} is the union of the double wedge with smaller angles and the rectangle $R$; see~\cref{fig:channel} (right). 
For a triangular channel the funnel of the channel can be obtained by prolonging the hypotenuse and the longest leg to rays emanating from the shared vertex. 
The funnel is the region bounded by the two rays that contains the channel.

Let $\tau$ be a line parallel to a shortest side of the channel, at distance $h$ from the channel 
and in particular, at distance $h$ from a shortest side of it. 
We denote as the \emph{shadow} at distance $h$ the length of the \emph{shadow segment}, which is the intersection of the funnel of the channel and $\tau$.
Figure~\ref{fig:channel} illustrates this description.

\begin{figure}
    \centering
    \includegraphics[page=4,width=\textwidth]{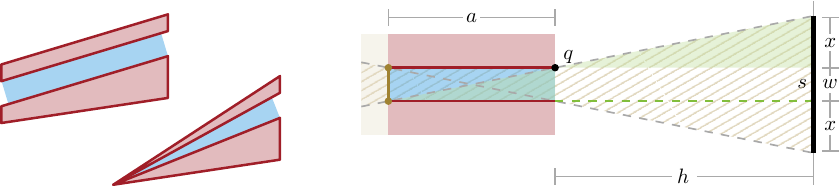}
    \caption{Left: channels (blue) of obstacles (red). Right: illustration of how a channel 
    restricts the visibility of any point placed on the brown line segment or left of it.
    The funnel of the channel is shown with stripes. Segment $s$ is the shadow segment at distance $h$.}
    \label{fig:channel}
\end{figure}

The following lemma will be useful to bound how narrow or long a channel needs to be for the construction.

\begin{lemma}
    \label{lem:channel}
    Given a channel $C$ with length $a$ and width $w$, with $a > w$,  
    the shadow at distance $h$ is at most $w\cdot\left(2\frac{h}{a} + 1\right)$.
\end{lemma}

\begin{proof}
Let $q$ be a corner of $C$
such that the shadow segment $s$ is at distance $h$ from it, 
and in the case of a triangular channel, also such that $q$ is not incident to the right angle of the channel. 
Refer to Figure~\ref{fig:channel} (right) for an illustration.
Consider the right triangle $T$ in the funnel of the channel with a vertex on $q$ and bounded by $s$, by the funnel boundary, and by a line through $q$ perpendicular to $s$. 
Let $x$ be the length of the side of $T$ opposite from $q$.
Triangle $T$ is similar to the channel in case it is triangular, and it is similar to any triangle defined by three corners of the channel in case it is rectangular. 
Therefore, by triangle similarity, $\frac{w}{a} = \frac{x}{h}$ and $x = \frac{wh}{a}$.
By construction, the length of the shadow segment $s$ is $2x+w = \frac{2wh}{a} + w = w\cdot\left(2\frac{h}{a} + 1\right)$ if the channel is rectangular, and 
$x+w = \frac{wh}{a} + w = w\cdot\left(\frac{h}{a} + 1\right)$ if the channel is triangular.
\end{proof}

\paragraph*{Grid cells}
We consider a $k \times k$ grid. 
Each grid cell $[i,j]$ can be thought of as having three layers, defined by (virtual) concentric squares, which we assume to be axis-aligned; see Figure~\ref{fig:reduction-cell}.%

\subparagraph{Inner layer.}
This layer, which we will refer to as the \emph{core}, is bounded by the innermost (virtual) square, which has a side length of $m^3$.
This square is subdivided into a (virtual) regular $m\times m$ subgrid, which we denote as the \emph{core grid}.
The core is where the correspondence to the subgrid in the \textsc{Grid Tiling} instance occurs. 
More precisely, in the core of grid cell $[i,j]$ there are $m^2$ relevant \emph{core regions} which are squares of side length $1$ arranged in the $m\times m$ subgrid, one per cell and placed centrally. 
We denote by \emph{core region} $(i',j')$ (of grid cell $[i,j]$) the one in the $i'$-th row and the $j'$-th column of the subgrid.
The core regions are in one-to-one correspondence with tiles in the \textsc{Grid Tiling} instance. 
A \emph{valid core region} is one that corresponds to a valid tile in $S_{i,j}$.

Note that there are no predrawn parts in the inner layer or around these regions, they arise from the restrictions imposed by other gadgets in the grid cell; see~\cref{fig:reduction-cell}.
    
\subparagraph{Middle layer.}
The middle layer is constructed around another (virtual) concentric square with side length $2\ceil{m^{7.5}} + m^3$.
On the four sides of the square there are four \emph{filtering vertices} $v_{i,j}^l$ (left), $v_{i,j}^t$ (top), $v_{i,j}^r$ (right), and  $v_{i,j}^b$ (bottom), labelled in clockwise order starting from the left side of the square. 
The {filtering vertices} are
at distance $m^{7.5} - m^6$
to the next clockwise square corner.

Other obstacles ensure that for grid cell $[i,j]$ the connections from a {filtering vertex} to a vertex in a core region is only possible if the region corresponds to a valid tile in $S_{i,j}$. 
We refer to these regions as \emph{valid} core regions.
The idea is to use channels to control connections between the {filtering vertices} and the core. 
This requires a careful construction of obstacles, as we will need that the only 
regions reachable from the four {filtering vertices} lie in valid core regions.
Moreover, we need to place the obstacles so that they do not block vertical and horizontal connections from the core regions. 
We call such set of obstacles the \emph{filtering gadget}.

We start by describing the bottom filtering gadget for the bottom {filtering vertex} $v_{i,j}^b$. 
The rest are analogous.
To describe the filtering gadget, we consider a (virtual) vertical line $r$, which is the supporting line of the left side of the (virtual) square 
bounding the core, at horizontal distance $m^6$ from $v_{i,j}^b$.

We denote as \emph{core grid point} $(i',j')$ with $(i',j') \in [m] \times [m]$, the central point of the $(i',j')$ core region.
For a core grid point $(i',j')$, we denote with $\rho_{i',j'}$ the ray emanating from $v_{i,j}^b$ containing $(i',j')$.
We consider only the rays emanating from $v_{i,j}^b$ and containing the center point of valid core regions. 
Thus, there are at most $m^2$ such rays.
For each such ray, let %
$q_{i',j'}$ a vertex at the intersection point with $r$.

For each $q_{i',j'}$, we create two vertices on $r$: 
one $\frac{1}{2m^3}$ units above $q_{i',j'}$ and one $\frac{1}{2m^3}$ units below $q_{i',j'}$. 
If the distance $d$ between two intersections $q_{i',j'}$ (above) and $q_{i'',j''}$ (below) was smaller than $\frac{3}{2m^3}$, 
instead of the vertices being placed $\frac{1}{2m^3}$ units below $q_{i',j'}$ and $\frac{1}{2m^3}$ units above $q_{i'',j''}$, 
we place them $d/3$ units below $q_{i',j'}$ 
and $d/3$ units above $q_{i'',j''}$, respectively.\footnote{This is never the case: it is possible to show that the distances between ray intersection points on $r$ are always larger than $\frac{3}{2m^3}$ for $m$ large enough. 
It follows from the fact that for $m$ large enough, the smallest possible angle between rays any two rays in our construction is the one formed by $\rho_{1,1}$ and $\rho_{2,1}$. 
But since the proof of this fact is rather technical, to simplify the presentation we use a construction whose validity that does not rely on this fact.}

We construct obstacle boundary pieces connecting $v_{i,j}^b$ with each $q_{i',j'}$.
Tracing along $r$, we close the obstacles by connecting consecutive points along $r$ without crossing any of the rays containing the center point of valid core regions. 
To close the topmost boundary piece, 
we connect the topmost $q_{i',j'}$ with a vertex $w_1$ with the same $x$-coordinate as $v_{i,j}^b$ 
and at distance $\ceil{m^{7.5}} - m^6$ to the left.
Similarly, to close the bottommost boundary piece, 
we connect the bottommost $q_{i',j'}$ with a vertex $w_2$ with the same $x$-coordinate as $v_{i,j}^b$ 
and at distance $m^6 + \frac{m^2-1}{2}$ to the left.
The two vertices $w_1$ and $w_2$ are connected to $v_{i,j}^b$. 

The obstacles in a filtering gadget define at most $m^2$ triangular channels of 
width smaller than $\frac{1}{m^3}$. 
We can bound the width of any funnel of a channel of the filtering gadget in the core region. 

\begin{lemma}
\label{lem:connection-shadow}
The funnel of a channel of the filtering gadget inside the core is contained in a strip of width $\frac{3}{m}$.
\end{lemma}
\begin{proof} 
We start by showing that a channel of a filtering gadget has length at least $m^6$. 
Consider a triangular channel of the bottom filtering gadget. 
Observe that any of its two long sides are at an angle of strictly more than $\frac{\pi}{4}$ 
with respect to the horizontal direction. 
Thus, by projecting its hypotenuse onto a horizontal line,
we obtain a length of at least $\sqrt{2}m^6$. 
The length of the channel is in fact slightly smaller than the length of the hypotenuse, but the difference is of order $O(\frac{1}{m^3})$.
Thus, for large enough $m$, the length of the channels is at least $m^6$ and their width is by construction smaller than $\frac{1}{m^3}$.

The distance from 
a {filtering vertex} to any point in a core region is upper bounded by $\sqrt{(m^6 + m^3)^2 + (\ceil{m^{7.5}} + m^3)^2} < m^8$ for $m$ large enough. 
By Lemma~\ref{lem:channel}, the shadow of the channel at distance $m^8$ is at most $\frac{1}{m^3} \cdot (2 \frac{m^8}{m^6} + 1) = \frac{2}{m} + \frac{1}{m^3} < \frac{3}{m}$ for $m>1$.
\end{proof}

There are two crucial properties of the channels defined by filtering gadgets. 

\begin{figure}
    \centering
    \includegraphics[page=6]{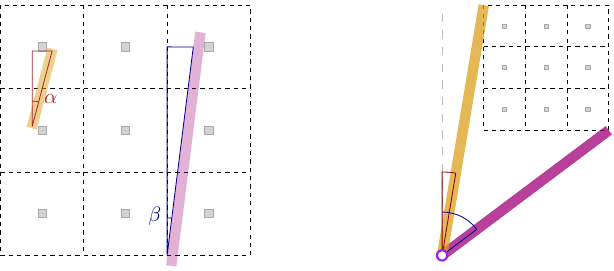}
    \caption{Illustration of the proofs of the properties of the filtering gadget.}
    \label{fig:connection-properties}
\end{figure}

\begin{lemma}
\label{lem:one-col}
  The funnel of a channel of a filtering gadget in a grid cell intersects exactly one row or one column of the core grid of the cell.
\end{lemma}
 
\begin{proof}
 We consider without loss of generality a bottom filtering gadget. 
 What we will in fact show is that the funnel of the channel defined by a valid tile 
 enters the core grid intersecting its bottom (and top) boundary in the column that corresponds to the valid tile.
 Looking at the core, it suffices to ensure that the slopes of the channels are not smaller than a certain threshold. 
 Equivalently, we will show that the angles defined by channels and a vertical line are always small enough. 
 
 Recall that, by Lemma~\ref{lem:connection-shadow}, the funnel of a channel of the filtering gadget on the core is contained in a strip of width $\frac{3}{m}$.
 To determine the slope threshold, we consider a strip of width $\frac{3}{m}$ 
 with the right side passing through the topmost leftmost point of a core region in the top row of the core grid 
 and with the left side passing through the bottommost leftmost point $p$ of the same core grid column; see Figure~\ref{fig:connection-properties} (purple, left).
 
 We define a right triangle $T$ bounded by the strip, by the left boundary of the core grid column, and by the line passing though the centers of the core regions on the top row. 
 Let $\beta$ denote the angle at $p$.
 Note that making the strip pass by through $p$ and by any other point of a core region would only result in an angle larger than $\beta$. 
 
 We now prove that the strips that contain the funnels of the channels of the filtering gadget define acute angles smaller than $\beta$, which implies the lemma statement.
 The vertical side of $T$ has length $(m-1)m^2 + \frac{m^2}{2}$.
 The horizontal of $T$ has length lower bounded by
 $\frac{m^2}{2} - \frac{3}{m} - \frac{1}{2}$. 
 Thus, $ \frac{\frac{m^2}{2} - \frac{3}{m} - \frac{1}{2}}{(m-1)m^2 + \frac{m^2}{2}} \le \tan \beta $. 
 
 We now consider the maximum acute angle $\beta^*$ formed with a vertical line by a strip of width $\frac{3}{m}$ bounding a gadget channel funnel, as in Lemma~\ref{lem:connection-shadow}. 
 We can get an upper bound for $\tan \beta^*$ by considering a strip though the {filtering vertex} and the bottom right corner of the core grid; see Figure~\ref{fig:connection-properties} (purple, right). 
 Thus, 
 for $m$ large enough, $\tan \beta^* < \frac{m^6 + m^3}{\ceil{m^{7.5}}} < \frac{\frac{m^2}{2} - \frac{3}{m} - \frac{1}{2}}{(m-1)m^2 + \frac{m^2}{2}} \le \tan \beta $. Therefore, since $0 \le \beta, \beta^* \le \frac{\pi}{2}$, it follows that $\beta^* < \beta$ as desired.
\end{proof}

\begin{lemma}
\label{lem:one-core}
  A channel of a filtering gadget in a grid cell intersects exactly one core region of the cell.
\end{lemma}

\begin{proof}
The proof follows, on a high level, similar lines to the proof of \cref{lem:one-col}.
We consider without loss of generality a bottom filtering gadget. 
Looking at the core, since \cref{lem:one-col} holds, it suffices to ensure that the slopes of the channels are not larger than a certain threshold. 
Equivalently, we will show that the angles defined by channels and a vertical line are always large enough. 

 Recall that, by Lemma~\ref{lem:connection-shadow}, the funnel of a channel of the filtering gadget on the core is contained in a strip of width $\frac{3}{m}$.
 To determine the slope threshold, we consider a strip of width $\frac{3}{m}$ 
 with the left side passing through the bottommost rightmost point $p$ of a core region that is not in the bottom row of the core grid 
 and with the right side passing through the topmost leftmost point $q$ of the core region in the core grid cell below; see Figure~\ref{fig:connection-properties} (yellow, left).

 We define a right triangle $T$ 
 with the top side aligned with the bottom of the core region containing $p$, 
 the bottom vertex at the same height as the top of the core region containing $q$,
 the right side through the middle of the strip, and the left side vertical.
 Let $\alpha$ denote the angle at the bottom vertex of $T$.
 Note that making the strip pass by any other points on two different core regions would only result in an angle smaller than $\alpha$. 
 
 We now prove that the strips that contain the funnels of the channels of the filtering gadget define acute angles larger than $\alpha$, which implies the lemma statement.
 The vertical side of $T$ has length $m^2 -1$.
 The horizontal side of $T$ has length upper bounded by $1+\frac{6}{m}$. 
 Thus, $\tan \alpha \le \frac{1+\frac{6}{m}}{m^2 -1} $. 
 
 We now consider the maximum acute angle $\alpha^*$ formed with a vertical line by a strip of width $\frac{3}{m}$ bounding a gadget channel funnel, as in Lemma~\ref{lem:connection-shadow}. 
 We can get a lower bound for $\tan \alpha^*$ by considering a strip though the {filtering vertex} and the top left corner of the core grid. 
 Thus, 
 for $m$ large enough, $\tan \alpha \le \frac{1 + \frac{6}{m}}{m^2 -1} \le \frac{m^6}{\ceil{m^{7.5}} + m^3} < \tan \alpha^*$. Therefore, since $0 \le \alpha, \alpha^* \le \frac{\pi}{2}$, it follows that $\alpha < \alpha^*$ as desired.
\end{proof}

The combination of the two previous lemmas makes the placement of the {filtering vertex} geometrically delicate. 
First, note that the core regions have to be sufficiently smaller than the core grid cells for the statements to hold.
Second, focusing without loss of generality on the bottom filtering vertex, $v_{i,j}^b$, 
the filtering vertex needs to be sufficiently non-centered to allow other vertical connections, but
the first lemma requires it to be sufficiently centered and far, 
and the second lemma adds conflicting constraints, preventing the funnels from being \emph{too} vertical.

Observe that by \cref{lem:one-col}, for any core grid point $(i,j)$, $\rho_{i,j}$ does not intersect any column of the core grid besides column $j$.
Thus, the rays $\rho_{i,j}$ appear around $v$ in lexicographic order,
first sorted by columns and then by rows.
Finally, note that the obstacles in the filtering gadget are well-defined and can be described with rational coordinates. 
To see that, note that the vertical distance of the intersection of ray $\rho_{i',j'}$ with $l$ and $r$ can be computed using right triangle similarity. %
The actual coordinates of the corners can be then computed by adding a rational number. 
We summarize it in the following observation: 

\begin{observation}
 \label{obs:filtering-gadget}
    The filtering gadget defines a set of non-overlapping obstacles describable with polynomial coordinates.   
\end{observation}

\subparagraph{Outer layer.}
The outer layer is contained in the final (virtual) square, with side length $2\ceil{m^{7.5}} + 2m^6 +  m^3 $ (each side is at distance $m^6$ to the corresponding side of the middle layer (virtual) square. 
The \emph{VH gadget} consists of orthogonal obstacles that force vertical or horizontal connections through the cell to pass by a row or column of core regions, respectively. 
Moreover, they force these connections to be roughly vertical or horizontal.

The obstacles are in the region formed by the square containing the outer layer and the one containing the middle one. 
The {VH gadget} consists of four large $L$-shaped obstacles and $4(m-1)$ narrow rectangles.
They create $4m$ channels in groups of $m$ on the top, bottom, left and right, centrally aligned with the $m$ rows and $m$ columns of core regions; see Figure~\ref{fig:reduction-solution}.
The sizes of the obstacles are chosen so that the channels have length $m^6$ and width $1/m^3$. 
Thus, the $4(m-1)$ narrow rectangular obstacles have length $m^6$ and width $m^2-1/m^3$.
The large $L$-shaped obstacles have sides of lengths $\ceil{m^{7.5}} + \frac{m^2}{2} - \frac{1}{2m^3}$, $m^6$, and $\ceil{m^{7.5}} + m^6 + \frac{m^2}{2} - \frac{1}{2m^3}$.
The four filtering vertices lie on the boundaries of the four large $L$-shaped obstacles, subdividing them. 

\begin{observation}
 \label{obs:outer-layer}
    The outer layer defines a set of non-overlapping obstacles describable with polynomial coordinates.   
\end{observation}

When making the construction for a full instance, we in general have cells that share a side
and thus obstacles that have a shared boundary (which we could remove); see \cref{fig:reduction-solution}.

\begin{lemma}
\label{lem:VHgadget}
    The funnel of the $i$-th top/bottom (left/right) channel of a VH gadget in a cell
    is contained in a vertical (horizontal) strip of width $3/m$ centered in the channel.
    In particular the funnel only intersects, 
    among the core regions, those on the $i$-th column (resp. row) and, 
    among the VH gadget channels on the opposite side, 
    the $i$-th bottom/top (right/left) channel. 
    Moreover, the intersection of the funnels of the $i$-th top/bottom channel and of the the $j$-th left/right channel is contained in the  $(i,j)$ core region.
\end{lemma}

\begin{proof}
    By construction, the channels in a VH gadget have length $m^6$ and width $\frac{1}{m^3}$. 
    Moreover, if we consider without loss of generality a horizontal left channel, the horizontal distance to a right channel is $2\ceil{m^{7.5}} + m^3$, which is smaller than $m^8$ 
    for $m$ large enough. 
    By Lemma~\ref{lem:channel}, the shadow of the channel at distance $m^8$ is, as in Lemma~\ref{lem:connection-shadow}, at most $\frac{1}{m^3} \cdot (2 \frac{m^8}{m^6} + 1) = \frac{2}{m} + \frac{1}{m^3} < \frac{3}{m}$ for $m>1$. 
    The core regions have side length $1$, which, if $m\ge 3$, is larger than the shadow width.
    Since the different channels on a side are separated at least a distance of $m^2 - \frac{1}{m^3}$, which is greater than $\frac{3}{m}$ for $m$ large enough, the first part of the statement follows.
    For the last part follows from the previous arguments and the alignment of the core regions.
\end{proof}

\subsection{\texorpdfstring{\boldmath Graph $H$}{Graph H}}
We now describe the graph $H$ whose edges we need to insert; see the black graph in~\cref{fig:reduction-solution}. 
For each $(i,j) \in [k] \times [k]$, $H$ contains:

\begin{itemize}
    \item A vertex $\xi_{i,j}$.
    \item Four edges connecting $\xi_{i,j}$ to the filtering vertices $v_{i,j}^l$, $v_{i,j}^t$, $v_{i,j}^r$, and  $v_{i,j}^b$.
    \item Edges connecting $\xi_{i,j}$ to $\xi_{i-1,j}$  and to $\xi_{i+1,j}$ (we refer to them as \emph{horizontal} edges), and edges connecting $\xi_{i,j}$ to $\xi_{i,j-1}$, and to $\xi_{i,j+1}$ (we refer to them as \emph{vertical} edges), 
    if such vertices are defined.
\end{itemize}

\subsection{Correctness}

\thmwonehard*
\begin{proof}[Proof of Theorem~\ref{thm:W1hard}]
We first show that given an instance of $k\times k$ Gird Tiling, 
for which we assume that the integer $m$ in it is large enough,
the constructed instance of \textsc{Straight-Line Planarity Extension} is valid and its size works for a parameterized reduction.
By construction and by \Cref{obs:filtering-gadget,obs:outer-layer}, for $m$ large enough the constructed predrawn graph defines a plane straight-line drawing of a graph of size $O(m^2 k^2)$. 
This bound follows from the fact that there are $k^2$ grid cells, each with $O(m^2)$ obstacles defined: $O(m^2)$ in the filtering gadgets and $O(m)$ in the outer layer. 

The predrawn vertices in the construction, by their definition and by Observations \ref{obs:filtering-gadget} and~\ref{obs:outer-layer}, have polynomially-bounded (in $m$ and $k$) rational coordinates, which can be easily transformed to polynomially-bounded integer coordinates. 
(We did not try to optimize the coordinate size, but instead focused on the clearness of the construction.)

The graph $H$ to be inserted is a connected graph of size $\Theta(k^2)$. 

We now show that the reduction is correct, that is, there is a solution to the \textsc{Grid Tiling} instance if and only if there is a solution to the \textsc{Straight-Line Planarity Extension}. 
See Figure~\ref{fig:reduction-solution} for an illustration.

Assume there is a solution to our instance of \textsc{Straight-Line Planarity Extension}.
Then, vertex $\xi_{i,j}$ must be connected using straight-line edges to the filtering vertices $v_{i,j}^l$, $v_{i,j}^t$, $v_{i,j}^r$, and $v_{i,j}^b$. 
Connections to a filtering vertex are only possible through channels of the filtering gadget or inside an obstacle. 
Since the four filtering vertices are not on the boundary of a common obstacle, 
$\xi_{i,j}$ must lie in cell $[i,j]$ and must be connected to the filtering vertices through channels of the respective filtering gadgets.
Thus, the horizontal and vertical edges of $H$ must pass through the channels of the VH gadget that connect adjacent cells. 
By \cref{lem:VHgadget}, $\xi_{i,j}$ must lie in a core region $(i',j')$. 
Moreover, the horizontal (vertical) connections must be through the $i'$-th right and left (bottom and top) channels. 
This implies that the vertices $\xi_{i,j}$ lie in core regions that  horizontally and vertically have the same index. 
More precisely, if we denote as $(x_{i,j},y_{i,j})$ the core region in which $\xi_{i,j}$ lies, then 
$x_{i,1} = x_{i,2} = \ldots = x_{i,k}$ for every $i \in [k]$ and $y_{1,j} = y_{2,j} = \ldots = y_{k,j}$ for every $j \in [k]$.

It remains to show that the core regions in which the vertices $\xi_{i,j}$ lie are valid. 
By \cref{lem:one-core}, a channel of a filtering gadget in a grid cell intersects exactly one core region of the cell, 
and, by contruction, it must be a valid one. 
Since $\xi_{i,j}$ must be connected to the filtering vertices through these channels, the core region in which it lies must be a valid one.
Thus, the core regions in which the vertices $\xi_{i,j}$ lie in the solution correspond to a valid solution to the \textsc{Grid Tiling} instance. 

For the other direction, assume that there is a solution to the \textsc{Grid Tiling} instance. 
We then can place the vertices $\xi_{i,j}$ in the core grid points according to the \textsc{Grid Tiling} solution. 
By our construction, the core regions in which these points lie are valid. 
Thus, there are channels that allow the connection between such a core grid point and the four filtering vertices in its cell. 
(Note that we cannot place the vertices $\xi_{i,j}$ arbitrarily in the core regions, as they might fall out of the funnels of gadgets that we need to use, but our construction guarantees that the core grid point is always in such funnels.)
Moreover, the horizontal and vertical alignment guaranteed by the validity of the solution allows us to draw all the horizontal and vertical edges of $H$ through channels of VH gadgets. 
We therefore can insert graph $H$ into our construction without intersections and using straight-line edges, and there is a solution to the \textsc{Straight-Line Planarity Extension} instance.
\end{proof}

\subsection{\texorpdfstring{\boldmath Partially predrawn (local) crossing number}{Partially predrawn (local) crossing number}}

Our reduction can be adapted to show W[1]-hardness of {\sc{Partially Predrawn Rectilinear Crossing Number}}, which we recall below.

\begin{tcolorbox}
{{\sc Partially Predrawn Rectilinear Crossing Number}\hfill \textbf{Parameter:} \(|E(G) - E(\Gamma)|\)}

\noindent\textbf{Input:} A straight-line partially predrawn graph \((G,\Gamma)\) where \(G\) is connected and an integer \(k\).\\
\noindent\textbf{Question:} Is there a geometric drawing extension of \(\Gamma\) to a straight-line drawing \(\mathcal{G}\) of \(G\) such that \(\crn(\mathcal{G}) - \crn(\Gamma) \leq k\)?
\end{tcolorbox}

The idea is to make the predrawn edges and obstacle boundaries \emph{thicker} by adding $k$ more parallel chains drawn closely together; see~\cref{fig:obstacles}. 
This effectively prevents obstacle boundaries from being crossed by edges of $G$ that need to be inserted. 
Moreover, we need to make $G$ connected, which we achieve by including in $\Gamma$ edges that connect the obstacles of the VH gadgets on the boundaries between cells. 
In this reduction, the crossing number would be $2k'^2-2k'$ where $k'$ is the grid size of the {\sc{Grid Tiling}} instance.
Observe that it is also easy to add predrawn edges to our construction (for example adding parallel subdivided paths to $\Gamma$ in obstacle channels) to force $k$ crossings, 
so in this case in the problem definition \(\crn(\mathcal{G}) - \crn(\Gamma) \leq k\) can also be replaced by \(\crn(\mathcal{G}) - \crn(\Gamma) = k\) and our hardness still would hold.
\begin{figure}[h]  
    \centering
    \includegraphics[page=5]{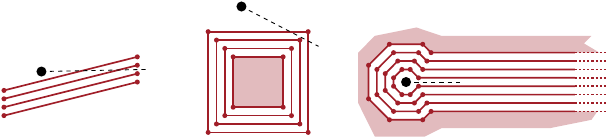}
    \caption{Illustration of predrawn edges and obstacle boundaries for $\ell = 3$.}
    \label{fig:obstacles}
\end{figure}

This same adaptation shows that {\sc{Partially Predrawn Straight-Line 1-Planarity}}, and in general {\sc{Partially Predrawn Straight-Line $k$-Planarity}} for any $k$, which we define below, is W[1]-hard. 
A graph is called \emph{$k$-planar} if it admits a \emph{$k$-planar drawing} which is one in which each edge is crossed at most $k$ times.
A closely related concept is the \emph{local crossing number} of a graph, which is the smallest nonnegative integer $k$ such that the graph is $k$-planar.
Using the adaptation above we can also show that {\sc{Partially Predrawn Rectilinear Local Crossing Number}} is W[1]-hard. 
For it we use that we can force some edges to have $k$ crossings as above.

\begin{tcolorbox}
{\sc{Partially Predrawn Straight-Line $k$-Planarity}\hfill \textbf{Parameter:} \(|E(G) - E(\Gamma)|\)}

\noindent\textbf{Input:} A straight-line partially predrawn graph \((G,\Gamma)\) where \(G\) is connected and an integer \(k\).\\
\noindent\textbf{Question:} Is there a geometric drawing extension of \(\Gamma\) to a straight-line drawing \(\mathcal{G}\) of \(G\) such that each edge inserted has at most $k$ crossings?
\end{tcolorbox}

\begin{tcolorbox}
{\sc{Partially Predrawn Rectilinear Local Crossing Number}\hfill \textbf{Parameter:} \(|E(G) - E(\Gamma)|\)}

\noindent\textbf{Input:} A straight-line partially predrawn graph \((G,\Gamma)\) where \(G\) is connected and an integer \(k\).\\
\noindent\textbf{Question:} Is there a geometric drawing extension of \(\Gamma\) to a straight-line drawing \(\mathcal{G}\) of \(G\) such that each edge inserted has at most $k$ crossings and such that there is no geometric drawing extension of \(\Gamma\) to a straight-line drawing \(\mathcal{G}\) of \(G\) such that each edge inserted has at most $k-1$ crossings?
\end{tcolorbox}

\section{Conclusion}
\label{sec:conclusion}
In this paper we presented a general framework and algorithm
to compute in \FPT-time the partially predrawn crossing number
under a wide set of drawing restrictions.
Crucially, our definition of topological crossing patterns allows us to
solve the crossing number problem for pseudolinear and fan-planar drawings,
resolving an open problem by M\"unch and Rutter~\cite{munch2024parameterized}.
We also showed that \textsc{Straight-Line Planarity Extension} is $W[1]$-hard when
parameterized by the number of edges that are not predrawn. 

A key problem we leave open is to resolve the parameterized complexity of deciding the rectilinear crossing number (without predrawings on which our and the previous hardness results strongly rely).
Beyond this, it would be interesting how much more general topological crossing patterns can be allowed to decide pattern-free crossing numbers in \FPT-time.
Notable directions that we do not yet know how to approach are weakening the side specification constraint, the crossing locality assumption, or the \(k\)-crossing contraction safeness we require of a set of forbidden topological crossing patterns.
Also, let us point out that the notion of drawing extension considered for our algorithmic result does not require orientations of disconnected components to carry over from the partial drawing.
This is in contrast to the setting in \cite{DBLP:conf/compgeom/HammH22}.
While it seems that there are no obvious obstacles in applying the flippability machinery from \cite{DBLP:conf/compgeom/HammH22}, this adds a layer of complexity that we decided not to include into the present work and carrying this out formally is left as future work.

\bibliography{cn-pseudolinear}

\end{document}